\documentclass[lettersize,journal,12pt,onecolumn]{IEEEtran}
\usepackage{amsmath,amsfonts}
\usepackage{algorithmic}
\usepackage{algorithm}
\usepackage{amssymb}
\usepackage{array}
\usepackage[caption=false,font=normalsize,labelfont=sf,textfont=sf]{subfig}
\usepackage{textcomp}
\usepackage{stfloats}
\usepackage{url}
\usepackage{xcolor}
\usepackage{verbatim}
\usepackage{graphicx}%导入图片
\usepackage{cite}
\usepackage{stfloats} %让跨双栏图片出现在正确的位置
\usepackage{caption}  %用了以后图片名字会居中

\newtheorem{definition}{Definition}
\newtheorem{theorem}{Theorem}

\newtheorem{remark}{Remark}

\begin{document}
	
\title{Nonlinear Information Theory: Characterizing Distributional Uncertainty in Communication Models with Sublinear Expectation}
	
\author{Wen-Xuan Lang, Shaoshi Yang, Jianhua Zhang, Zhiming Ma

\thanks{Wen-Xuan Lang and Zhiming Ma are with the National Center for Mathematics and Interdisciplinary Sciences, Academy of Mathematics and Systems Science, Chinese Academy of Sciences, Beijing 100190, China (e-mail: langwx@amss.ac.cn; mazm@amt.ac.cn).}
\thanks{Shaoshi Yang is with the School of Information and Communication Engineering, Beijing University of Posts and Telecommunications, with the Key Laboratory of Universal Wireless Communications, Ministry of Education, and also with the Key Laboratory of Mathematics and Information Networks, Ministry of Education, Beijing 100876, China (e-mail: shaoshi.yang@bupt.edu.cn).}
\thanks{Jianhua Zhang is with the School of Information and Communication Engineering, Beijing University of Posts and Telecommunications, and also with the State Key Laboratory of Networking and Switching Technology, Beijing 100876, China (e-mail: jhzhang@bupt.edu.cn).}

}
\maketitle

\begin{abstract}
A mathematical framework for information-theoretic analysis is established, with a new viewpoint of describing transmitted messages and communication channels by the nonlinear expectation theory, beyond classical probability framework. Based on nonlinear expectation theory, in this paper we define several fundamental concepts, such as nonlinear information entropy, nonlinear joint entropy, and nonlinear mutual information, and establish their basic properties. By using strong law of large numbers under sublinear expectations, we propose a nonlinear source coding theorem, which shows that nonlinear entropy is a sufficient achievability threshold for the achievable coding rate of uncertain-distribution sources under the maximum error probability criterion, and determines a cluster point of the achievable coding rate under the minimum error probability criterion. We also propose a nonlinear channel coding theorem, providing explicit expressions of the upper bound and a cluster point under the maximum and minimum error probability criterion, respectively, for the achievable coding rate of uncertain-distribution channels. Additionally, we propose a nonlinear rate-distortion source coding theorem, proving that the rate-distortion function based on nonlinear mutual information is a cluster point of the lossy compression performance of uncertain-distribution sources under the minimum expected distortion criterion. Finally, we show some examples and simulations to consolidate our theoretic study.
\end{abstract}
	
	
\section{Introduction}\label{sec1}
\subsection{Motivation}
The notion of uncertainty plays a central role in information theory, which is built on the theory of probability and stochastic processes \cite{book}. By assuming specific individual distributions for information sources, communication channels, noises, and interferences, the key concepts of Shannon information theory, such as mutual information, entropy, and channel capacity, are defined. The classical Shannon information theory can be viewed as only considering the uncertainty of one degree, i.e., the probability distributions are defined for random events only and the probability space defined with the measure theory is deterministic and explicitly known. In other words, the probability models in classical Shannon information theory are assumed to have no uncertainty.

However, in the complex physical world fraught with a significant amount of unanticipated and nonstationary random phenomena, the assumption of precise and well-defined probability distributions to describe uncertainties is somewhat idealized. In many cases that, for instance, have only limited observations or small sample size, there may not be a single, fixed probability distribution that accurately captures the underlying randomness. This reality challenges the traditional approach of relying on deterministic probability models, which typically assume the existence of an underlying ``true'' probability distribution. As pointed out in \cite{spiegelhalter2024does}, probability is not an objective property of the world but rather a construct based on subjective judgments and assumptions. In most practical scenarios, the assessments of probability are not estimates of some ``true'' probability but rather expressions of personal or collective view of uncertainty. This perspective underscores the need of developing more flexible and realistic methods to characterize uncertainty, moving beyond the limitations of assuming fixed and known probabilistic distributions.

Given these insights, it is important to recognize that the classical Shannon information theory can be interpreted as operating within a linear probabilistic framework, where uncertainty is quantified through linear expectations and moment-based characterizations. Although the traditional probability theories and methods can help us understand the laws and structures of random systems to some degree, we are often unable to completely eradicate the higher-level uncertainties inherent in the probability models themselves that underlie these methods. This is particularly the case when the method adopted does not match the underlying probability model studied. 

An important question naturally arises: Can the classical Shannon information theory be extended to characterize scenarios that have multiple levels of uncertainty or \textit{uncertainty of uncertainty}? If yes, how to extend the classical Shannon information theory to such scenarios? These are the two major motivational questions behind this study.

At the time of writing, the telecommunications community has begun to focus on the potential candidate technologies for 6G \cite{chowdhury20206g,tataria20216g,nikitopoulos2022massively,zhu2023pushing, wang2024clutter}. The higher precision requirements for various algorithms associated with  information sources and communication channels in 6G require us to consider the impact of uncertainty in the probability model itself as well. Therefore, new fundamental theories, as well as the related strong and trustworthy analysis, are much needed. In light of this, it is imperative to consider a more holistic landscape of uncertainties in wireless communications.

\subsection{Related Work in Information Theory}
Traditionally, there are two types of uncertainty affected by the probability models used in communication systems. The first is the uncertainty inherent in information sources. Information entropy, a well known concept for measuring uncertainty in information sources, derives from the linear expectation of a deterministic probability space. Among various definitions of information entropy, the most famous one, Shannon entropy \cite{shannon1948mathematical}, provides a mathematical tool for quantifying information and analyzing the efficiency of data compression. It plays a key role in solving practical problems and helps people deeply understand the essence of information and process data.  The successful application of Shannon entropy has prompted researchers to extend it from various aspects, leading to the emergence of concepts that can be more generalized than Shannon entropy. More specifically, R{\'e}nyi entropy, which was introduced by R{\'e}nyi \cite{renyi1961measures}, can be regarded as an extension of Shannon entropy. Marichal first established the notion of Choquet capacity entropy \cite{marichal2002entropy}, which is fundamentally dependent on Choquet capacity and Choquet expectation. There are also many other works devoted to the generalization of entropy based on the deterministic probability distribution, such as \cite{sharma1975entropy,rao2004cumulative,tsallis2011nonadditive,nielsen2011closed}. However, for all these methods, it is crucial to provide an explicit expression for the probability model, regardless of using the Choquet capacity or the probability measure. As a result, bias or error may be imposed on signal or information processing algorithms once the real distribution of the information source diverges from the estimated distribution.

Furthermore, in source coding, Ziv proposed fundamental concepts and frameworks for universal source coding \cite{ziv1972coding1,ziv1972coding2}, and a universal algorithm for sequential data compression was presented \cite{ziv1977universal}. A universal Wyner–Ziv coding setup was introduced in \cite{jalali2010universal}, where the distribution of the source is assumed unknown while the conditional distribution of the side-information channel is perfectly known. In \cite{dupraz2012source}, the authors used a single distribution chosen from a family of candidate distributions to describe the information sources, and the distribution chosen varies with a set of given parameters that either change from symbol to symbol or remain constant for a while. To elaborate a little further, for the theoretical analysis of universal sources, the open literature typically assumes the existence of a particular distribution characterizing the random variable, even though the specific parameters and/or form of this distribution are uncertain. In other words, although universal source coding can directly encode sources with uncertain or unknown distributions, in theoretical analysis, the existing methods still assume that the underlying probability model (i.e., the probability space) is deterministic.

In essence, traditional research on information sources is restricted to analyzing the uncertainty characterized by random variables and the corresponding well-defined probability distributions, and is incapable of dealing with the uncertainty inherent in probability models themselves. Therefore, it is difficult to use traditional theoretical framework to support dynamic compatibility with various individual sources, hybrid sources, and unknown sources. This limitation may result in a high degree of inaccuracy when modeling the physical law dominated but still randomly changing complex real world. Instead of assuming the existence of a fixed probability distribution, it makes more sense to recognize that the underlying randomness may not be captured by any deterministic probability model. Universal source coding is a powerful and elegant solution for compressing data without knowing the exact source distribution. Our work aims to use nonlinear expectation theory to directly model and analyze distributional uncertainty, without relying on any single underlying distribution. This approach does not conflict with universal source coding, but instead helps to strengthen its theoretical basis in handling distributional uncertainty.

The second is the uncertainty inherent in communication channel. The mathematical model used to characterize the uncertainty of the channel is typically described as a deterministic transition probability matrix. In 1959, Blackwell et al. first generalized Shannon's basic theorem on the capacity of a channel to a class of memoryless channels \cite{blackwell1959capacity}, which can be viewed as accounting for the uncertainty of the transition probability matrix itself. Afterwards, the concept of arbitrarily varying channels (AVC) was established \cite{blackwell1960capacities} in order to simulate communication channels with unknown transition probability matrices that may vary with time in an arbitrary manner. An arbitrarily varying channel can be characterized by a family of transition probability matrices $\mathcal{C}=\{ \boldsymbol{P}_s(\cdot|\cdot),s \in \mathcal{S} \}$, where $s\in \mathcal{S}$ is the index of a specific transition probability matrix $\boldsymbol{P}$ in the set $\mathcal{C}$, under the classical communication framework. The definition of random coding capacity was proposed in \cite{ahlswede1969correlated} and derived for AVC under the maximum error probability or average error probability criterion \cite{ahlswede1970note,ahlswede1978elimination}. In \cite{ahlswede1980method}, Ahlswede presented a strong theoretical performance guarantee that AVC can achieve the random coding capacity in some situations under the maximum error probability criterion. In recent years, there is a renewed interest in studying the impacts imposed by a family of transition probability matrices, such as \cite{chen2010tighter,yang2018intrinsic,vu2020uncertainty}.

All the previous works based on the concept of a family of transition probability matrices actually still assumed that the channel follows a certain distribution of this family within a specific observation slot in time, frequency or spatial domain. In other words, the probability model remains conditionally deterministic, but its specific expression may not be known to the encoder or decoder. Existing research on channels with a family of transition probability matrices has made significant and impressive progress, providing valuable insights and practical solutions. Inspired by these advancements, we also aim to take a complementary theoretical approach by directly studying channels whose distributions are uncertain through the framework of nonlinear expectation theory.

\subsection{Novel Contributions}
Against the above background, in this paper, we adopt nonlinear expectation theory \cite{peng2007g} as a mathematical framework for extending classical information theory to scenarios with distributional uncertainty. This framework allows uncertainty to be characterized at the level of probability models for information sources and communication channels. Stimulated by the idea of nonlinear expectation theory, in this paper we propose a \textit{nonlinear information theory}, by considering a new nonlinear communication model, which has multi-level uncertainties in probability distributions. The main contributions can be summarized as follows:
\begin{enumerate}
	\item{We generalize several fundamental concepts related to information sources and entropy. We deduce a nonlinear information entropy as the form of
		\begin{equation*}
		\hat{H}(X):=\sup_{\theta \in \Theta } \sum _{x}p_{\theta }(x)\log \frac{1}{p_{\theta }(x)},
		\end{equation*}
		from a set of new hypotheses. Nonlinear joint entropy, nonlinear conditional entropy and nonlinear mutual information are also explicitly defined. Some fundamental properties, such as the chain rule and the Fano inequality, are generalized as well.}
	
	\item{We propose a nonlinear source coding theorem. Specifically, based on the strong law of large numbers under sublinear expectations \cite{zhang2023sufficient}, we show that the nonlinear information entropy $\hat{H}(X)$ is a sufficient achievability threshold for the achievable coding rate of uncertain-distribution sources under the maximum error probability criterion, and also determines a cluster point $\underset{\theta \in \Theta}{\inf} \sum _{x}p_{\theta }(x)\log \frac{1}{V(x)}$ of the achievable coding rate of uncertain-distribution sources under the minimum error probability criterion.}
	
	\item{We propose a nonlinear channel coding theorem. Suppose $\mathcal{X}$ is the input alphabet and $\mathcal{Y}$ is the output alphabet. The uncertain-distribution channel model is potentially characterized by a family of transition probability matrices $\left\{\boldsymbol{P}_{\lambda } \in [0,1]^{\mathcal{X}\times \mathcal{Y}}\right\}_{\lambda \in \Lambda }$. After denoting $R_{\textrm{c}}$ as the achievable channel coding rate of communication systems that rely on the proposed uncertain-distribution channel model, the upper bound and a cluster point of $R_{\textrm{c}}$ for any given scenario are determined under the maximum error probability criterion and the minimum error probability criterion, respectively.}
	
	\item{We also propose a generalized nonlinear rate-distortion coding theorem. We characterize the lossy compression for uncertain-distribution sources, and establish the concept of maximum expected distortion and minimum expected distortion. On this basis, we prove that the rate distortion function based on the nonlinear mutual information is a cluster point of the lossy compression performance of uncertain-distribution sources under the minimum expected distortion criterion.}
\end{enumerate}

It is worth clarifying that the terms such as ``nonlinear information theory'' and ``nonlinear information entropy'' used in this paper refer to the mathematical framework of sublinear expectations for modeling distributional uncertainty, and should not be confused with the notion of ``nonlinear channels'' that arises in other contexts, such as optical fiber communications, where nonlinearities are caused by physical media like optical fibers or devices like amplifiers/quantizers. The two are conceptually distinct and should not be confused with each other.

Nonlinear expectation theory provides a natural and fundamental mathematical framework for modeling distributional uncertainty in this work. Being the first to study information theory from the nonlinear expectation theory perspective, this work represents a paradigm shift, extending the classical information theory and communication models from linear to nonlinear. This new framework not only accommodates the complex randomness of real-world scenarios but also provides novel insights into the behavior of information sources and communication systems under higher-level of uncertainty.

\subsection{Organization}
The remainder of the paper is organized as follows. In Section \ref{sec2}, we briefly review the basics of nonlinear expectation theory. In Section \ref{sec3}, we introduce a nonlinear communication model and formulates uncertain-distribution sources and channels, together with the corresponding definitions of nonlinear information measures. In Section \ref{sec4}, the main results of our nonlinear information theory are presented, including fundamental properties and coding theorems for uncertain-distribution sources and channels under different performance criteria. In Section \ref{sec5}, we show some examples and applications, and present simulation results to consolidate our theoretic study. The conclusions are drawn in Section \ref{sec6}. Additionally, in Appendix \ref{app1}, we provide an intuitive explanation of why we use the nonlinear expectation theory and compare it with traditional probabilistic methods to highlight the differences and advantages, and the proofs of the main results are deferred to the appendices \ref{app2}-\ref{app5}.

\section{Basic Notions of Nonlinear Expectation Theory}\label{sec2}
Peng \cite{peng2007g} established the nonlinear expectation theory, which extends classical probability theory by using sublinear expectations instead of linear expectations. This theory is an emerging field of mathematics and has offered significant benefits to practical applications such as finance and risk management \cite{peng2023improving}. The central idea of this theory is that random variables have uncertainty in the distribution itself. By incorporating the distribution uncertainty, this theory provides a more  general framework for modeling complex stochastic systems. In this section, we briefly introduce the major concepts of nonlinear expectation theory. For details of the following definitions and theorems, readers can refer to \cite{peng2019nonlinear, peng2017theory}.

Let $\Omega$ be a sample set. Let $\mathcal{H}$ be a linear space of real valued functions defined on $\Omega$. Suppose that $\mathcal{H}$ satisfies the following three conditions:

1. $c\in \mathcal{H}$ for each constant $c \in \mathbb{R}$.

2. $|X|\in \mathcal{H}$ if $X \in \mathcal{H}$.

3. $\phi(X_1,\cdots, X_n)\in \mathcal{H}$ for each $\phi \in \mathbb{L}^{\infty}(\mathbb{R}^n)$ if $X_1,\cdots, X_n\in \mathcal{H}$.

\noindent Here $\mathbb{L}^{\infty}(\mathbb{R}^n)$ denotes the space of bounded Borel-measurable functions on $\mathbb{R}^n$.

\begin{definition}\label{def:sublinear expectation}
	A \textit{sublinear expectation} $\mathbb{E}:\mathcal{H} \rightarrow \mathbb{R}$ is a functional defined on the space $\mathcal{H}$ satisfying the following properties:
	\begin{enumerate}
		\item{\textit{Monotonicity}: $\mathbb{E}[X]\geq \mathbb{E}[Y]$, if $X\geq Y$.}
		\item{\textit{Constant preserving}: $\mathbb{E}[c]=c,\forall c \in \mathbb{R}$.}
		\item{\textit{Sub-additivity}: $\mathbb{E}[X+Y]\leq \mathbb{E}[X]+\mathbb{E}[Y]$, $\forall X,Y \in \mathcal{H}$.}
		\item {\textit{Positive homogeneity}: $\mathbb{E}[\lambda X]=\lambda \mathbb{E}[X]$, for $\lambda \geq 0$.}
	\end{enumerate}
	The triplet $(\Omega,\mathcal{H},\mathbb{E})$ is called a \textit{sublinear expectation space}. If $\mathbb{E}$ satisfies only 1) and 2), then $\mathbb{E}$ is called a nonlinear expectation and $(\Omega,\mathcal{H},\mathbb{E})$ is called a nonlinear expectation space.	
\end{definition}

\begin{remark}
	Given a sublinear expectation $\mathbb{E}$, the conjugate expectation $\mathcal{E}$ of sublinear expectation $\mathbb{E}$ is defined as
	\begin{equation}\label{eq:conjugate_expectation}
	\mathcal{E}[X]:=-\mathbb{E}[-X],\quad \forall X\in \mathcal{H}.
	\end{equation}
\end{remark}

\begin{theorem}\label{thm:representation theorem}
	Let $\mathbb{E}$ be a sublinear expectation defined on $(\Omega,\mathcal{H})$. Then there exists a family of linear functionals $\{E_{\theta}:\theta \in \Theta \}$ defined on $\mathcal{H}$, such that
	\begin{equation*}
	\mathbb{E}[X]=\sup_{\theta \in \Theta }E_{\theta}[X], \quad X\in \mathcal{H}.
	\end{equation*}
	If $\mathbb{E}$ also satisfies $\mathbb{E}[X_{i}]\downarrow 0$ for each sequence $\{ X_{i}\}_{i=1}^{\infty}$ of random variables in $\mathcal{H}$ such that $X_{i}\downarrow 0$, then for each $\theta \in \Theta$, there exists a probability measure $P_{\theta}$ defined on the measurable space $(\Omega,\sigma(\mathcal{H}))$ such that
	\begin{equation*}
	E_{\theta}[X]=\int_{\Omega} XdP_{\theta}, \quad \forall X\in \mathcal{H}.
	\end{equation*}
	Then $E_{\theta}$ is also denoted as $E_{P_{\theta}}$ and $\mathbb{E}[\cdot]=\underset{\theta \in \Theta}{\sup} E_{P_{\theta}}[\cdot]$. Here ``$ X_i \downarrow 0$ " means $X_i$ monotonically decreases to zero, and $\sigma(\mathcal{H})$ is the smallest $\sigma$-algebra generated by $\mathcal{H}$.
\end{theorem}

In the remainder of this paper, we exclusively focus on the sublinear expectation $\mathbb{E}$ satisfying the conditions in Theorem \ref{thm:representation theorem} unless otherwise specified. 

\begin{remark}
	We call the family $\mathcal{P}$ the uncertain probability measures associated with the sublinear expectation $\mathbb{E}$. For a given $n$-dimensional random variable $X$ defined on a sublinear expectation space $(\Omega,\mathcal{H},\mathbb{E})$, the probability measure family $\mathcal{P}$ gives rise to a family of probability distributions $\{ F_X(\theta,A)=P_{\theta}(X\in A) , \ A\in \mathcal{B}(\mathbb{R}^n) \}_{\theta\in \Theta}$, where $\mathcal{B}(\mathbb{R}^n)$ is the Borel $\sigma$-algebra on $\mathbb{R}^n$. For notational convenience, we also denote this family of probability distributions as $\{p_{\theta}(X)\}_{\theta \in \Theta}$. In such a case, we claim the corresponding uncertain probability distributions of $X$ are $\left\{p_{\theta }(X)\right\}_{\theta \in \Theta }$.
\end{remark}

In this new framework, the nonlinear versions of the notions of independence and identical distribution play important roles, and these notions are shown in the following two definitions.

\begin{definition}\label{def:identically distributed}
	Let $X_1$ and $X_2$ be two $n$-dimensional random variables defined on sublinear expectation spaces $(\Omega,\mathcal{H},\mathbb{E}_1)$ and $(\Omega,\mathcal{H},\mathbb{E}_2)$, respectively. They are called \textit{identically distributed}, denoted by $X_1\overset{d}{=}X_2$, if for any $\phi \in \mathbb{L}^{\infty}(\mathbb{R}^n)$ we have
	\begin{equation*}
	\mathbb{E}_1[\phi(X_1)]=\mathbb{E}_2[\phi(X_2)].
	\end{equation*}
	We say that the distribution of $X_1$ is stronger than that of $X_2$ if $\mathbb{E}_1 [\phi(X_1 )]\geq \mathbb{E}_2 [\phi(X_2 )]$, for each $\phi\in \mathbb{L}^{\infty}(\mathbb{R}^n)$.
\end{definition}

\begin{definition}\label{def:independent}
	In a sublinear expectation space $(\Omega,\mathcal{H},\mathbb{E})$, a random variable $Y$ is said to be \textit{independent} of another random variable $X$ under $\mathbb{E}[\cdot]$, if for any $\phi \in \mathbb{L}^{\infty}(\mathbb{R}^2)$ we have
	\begin{equation}\label{eq:def_independence}
	\mathbb{E}[\phi(X,Y)]=\mathbb{E}[\mathbb{E}[\phi(x,Y)]_{x=X}].
	\end{equation}
\end{definition}

It is important to observe that, under a sublinear expectation, the statement ``$Y$ is independent of $X$'' does not in general imply that ``$X$ is independent of $Y$''. In fact, cases of asymmetric independence between random variables in the real world are more common than cases of symmetric independence \cite{manner2010testing,vinod2022asymmetric}. This is especially true for random variables related to time order. Although random variables may not satisfy traditional independence in practice, it is possible that they satisfy the independence under sublinear expectation.

Motivated by the notion of independent and identically distributed in a sublinear expectation space, we also adopt the concept of IID random variables sequence under sublinear expectation, which is shown in Definition \ref{def:IID} below.

\begin{definition}\label{def:IID}
	Suppose	that $\{X_{1},X_{2},\cdots,X_{n},\cdots\}$ is a sequence of random variables defined on a sublinear expectation space $(\Omega,\mathcal{H},\mathbb{E})$.
	\begin{enumerate}
		\item {$X_{1},X_{2},\cdots,X_{n},\cdots$ are said to be \textit{identically distributed} if for each pair of $X_{i},X_{j},i \neq j$ and for each Borel-measurable function $\phi$ such that $\phi(X_{i}),\phi(X_{j})\in \mathcal{H}$, there holds
			\begin{equation*}
			\mathbb{E}[\phi(X_{i})]=\mathbb{E}[\phi(X_{j})].
			\end{equation*}}
		\item {Random variable $X_{n}$ is said to be \textit{independent} to $\boldsymbol{X}:=(X_{1},\cdots,X_{n-1})$ under sublinear expectation $\mathbb{E}$ if for each Borel-measurable function $\phi$ on $\mathbb{R}^{n}$ with $\phi(\boldsymbol{X},X_{n})\in \mathcal{H}$ and $\phi(\boldsymbol{x},X_{n})\in \mathcal{H}$ for each $\boldsymbol{x}\in \mathbb{R}^{n-1}$, we have
			\begin{equation*}
			\mathbb{E}[\phi(\boldsymbol{X},X_{n})]=\mathbb{E}\left[\mathbb{E}[\phi(\boldsymbol{x},X_{n})]_{\boldsymbol{x}=\boldsymbol{X}}\right].
			\end{equation*}}
		\item {A sequence of random variables $\{X_{1},X_{2},\cdots,X_{n},\cdots\}$ is said to be \textit{IID} if $X_{1},X_{2},\cdots,X_{n},\cdots$ are identically distributed and $X_{i}$ is independent to $\boldsymbol{X}:=(X_{1},\cdots,X_{i-1})$ for each $i \geq 1$.}
	\end{enumerate}
\end{definition}

\begin{definition}\label{def:capacity}
	A pair of capacities $(V,v)$ is said to be generated by a sublinear expectation $\mathbb{E}$, if
	\begin{equation*}
	V(A):=\mathbb{E}[I_{A}],\quad v(A):=-\mathbb{E}[-I_{A}], \quad \forall A \in \mathcal{F},
	\end{equation*}
	where $I_{A}$ is the indicator function of event $A$, and $\mathcal{F}$ is a $\sigma$-algebra.
\end{definition}

Let
\begin{equation*}
C_V(X):=\int_{0}^{\infty}V(X\geq t)dt+\int_{-\infty}^{0}[V(X\geq t)-1]dt.
\end{equation*}
A conclusion from the strong law of large numbers under sublinear expectations and is instrumental in developing our nonlinear information theory is presented as follows.

\begin{theorem}\label{thm:SLLN}
	Let $\{X_i\}_{i=1}^{\infty}$ be a sequence of IID random variables defined on a sublinear expectation space $(\Omega,\mathcal{H},\mathbb{E})$ where $\mathbb{E}[\cdot]=\underset{\theta \in \Theta}{\sup} E_{\theta}[\cdot]$. Suppose that $C_V(|X|)< \infty$, $\{P_{\theta}\}_{\theta\in \Theta}$ is a countably-dimensional weakly compact family of probability measures on $(\Omega,\sigma(\mathcal{H}))$ in the sense that, for any bounded $Y_1,Y_2,\cdots \in \mathcal{H}$ and any sequence $\{P_n\}\subset \{P_{\theta}\}_{\theta\in \Theta}$, there is a subsequence $\{n_k\}$ and a probability measure $P\in \{P_{\theta}\}_{\theta\in \Theta}$ such that
	\begin{equation*}
	\lim_{k\rightarrow\infty} P_{n_k}(\phi(Y_1,\cdots,Y_d)) =P(\phi(Y_1,\cdots,Y_d)).
	\end{equation*}
	Denote by $C(\{x_n\})$ the set of cluster points of the sequence $\{x_n\}$. Then, for any $b\in [-\mathbb{E}[-X_1],\mathbb{E}[X_1]]$ we have
	\begin{equation*}
	V\left(b\in C\left(\left\{\frac{\sum_{i=1}^{n}X_i}{n}\right\}\right)\right) = 1.
	\end{equation*}
\end{theorem}

\section{Establishment of Nonlinear Information Theory}\label{sec3}

The newly developed nonlinear expectation theory can be utilized to analyze the uncertainties of probability models themselves in the complex physical world. This motivates us to consider characterizing the communication model on a sublinear expectation space, resulting in a nonlinear communication model with uncertain distributions, as shown in Fig. \ref{fig_1}.

\begin{figure}[htbp]
	\centering
	\includegraphics[scale=0.7]{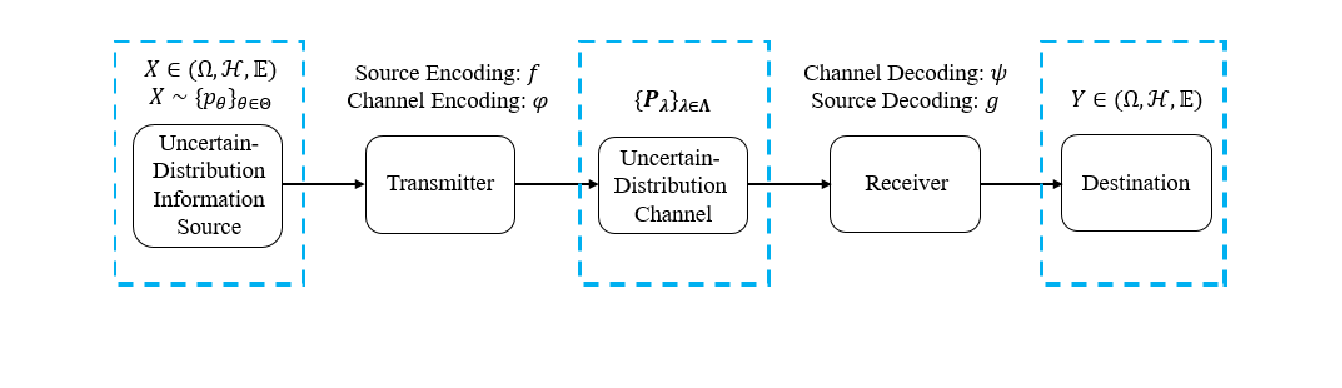}
	\captionsetup{font={scriptsize}}
	\caption{Nonlinear communication model with uncertain distributions.}
	\label{fig_1}
\end{figure}

Compared with the classical communication model, in the proposed nonlinear communication model, we assume that the random variables describing the messages sent by information sources are defined on the sublinear expectation space $(\Omega,\mathcal{H},\mathbb{E})$. In this case the information source is referred to as \textit{uncertain-distribution information source}. The implicit meaning is that there are higher-level uncertainties inherent in the distributions of the random variables describing information sources, which do not obey a deterministic probability distribution. We also assume that the interference and noise that affect signal transmission over the communication channel are random variables defined on a sublinear expectation space, meaning that the transition probability matrices of communication channels are also uncertain. In this case the communication channel is called \textit{uncertain-distribution channel}. In summary, there are two types of uncertainties regarding the probability models involved in the transmission process $X \rightarrow Y$ according to our nonlinear communication model. One is the uncertainty of the probability distributions of $X$ and the other is the uncertainty of the transition probability matrices of $X \rightarrow Y$.

For uncertain-distribution information sources, we aim to consider the amount of information it contains from the standpoint of the circumstance with the highest uncertainty. Let $\hat{H}(X)$ denote the measurement of the amount of information contained in $X$, where $X$ is a discrete random variable on a sublinear expectation space $(\Omega,\mathcal{H},\mathbb{E})$ and the corresponding uncertain probability distributions of $X$ are characterized by $\left\{p_{\theta }(X)\right\}_{\theta \in \Theta }$. In order to reduce the mismatching error between the probabilistic model of a system and the practical system itself, it is crucial to carefully analyze a series of random variables that have unknown distributions. Especially, if the family of uncertain probability distributions of a random variable contains the uniform distribution, then the uniform distribution should be used for analyzing the amount of information embedded in the random variable, because uniform distribution represents the highest uncertainty case. Furthermore, it makes intuitive sense that the amount of information embedded in a random variable increases with the size of the uncertain probability distributions family the random variable may follow. In light of the above discussions, it is reasonable to make the following assumptions:

\textit{Assumptions:}
\begin{enumerate}
	\item{$\hat{H}(X)\leq \underset{\theta \in \Theta }{\sup } \sum _{x}p_{\theta }(x)\log \frac{1}{p_{\theta }(x)}$ and $\hat{H}(X)$ is continuous with each distribution in the set $ \left\{p_{\theta }(X)\right\}_{\theta \in \Theta }$. }
	\item{If the number of states for $X$ is $N$ and the uniform distribution $p(i)=\frac{1}{N},i=1,2,\cdots,N$, is a member of the distribution family $\{p_{\theta }(X)\}_{\theta \in \Theta }$, then $\hat{H}\left(X\right)=\log N$.}
	\item{Let $X$ and $Y$ be two random variables defined on the sublinear expectation spaces $(\Omega,\mathcal{H},\mathbb{E}_1)$ and $(\Omega,\mathcal{H},\mathbb{E}_2)$, respectively. If the distribution of $Y$ is stronger than that of $X$, then $\hat{H}(X)\leq \hat{H}(Y)$.}
	\item{If a random choice\footnote{Here the meaning of ``random choice" can be understood as an action that has multiple possible outcomes, as detailed in Section 6 of \cite{shannon1948mathematical}.} can be broken down into two successive random choices, the original $\hat{H}$ corresponding to the single random choice should be no greater than the supremum of the weighted sum of individual values of $\hat{H}$  corresponding to the two successive random choices.}
\end{enumerate}

Based on the above four assumptions, we derive the following theorem, which presents the mathematical expression of the amount of information contained in discrete random variables defined on a sublinear expectation space. Furthermore, in Section \ref{sec4}, we will show that several important properties can be attained by merely relying on these four fundamental assumptions.

\begin{theorem}\label{theentropy}
	Let $X$ be a discrete random variable on a sublinear expectation space $(\Omega,\mathcal{H},\mathbb{E})$. The corresponding uncertain probability distributions of $X$ are $\left\{p_{\theta }(X)\right\}_{\theta \in \Theta }$. Then the amount of information contained in $X$ can be expressed as $\hat{H}\left(X\right)=\underset{\theta \in \Theta }{\sup } \sum _{x}p_{\theta }(x)\log \frac{1}{p_{\theta }(x)}$.
\end{theorem}

\textit{Proof}: See Appendix \ref{app2}. $\hfill\blacksquare$

Without loss of generality, in this paper we consider only discrete random variables. Then, based on Theorem \ref{theentropy}, the nonlinear information entropy of discrete random variables is formally defined below. 

\begin{definition}
	Let $X$ be a discrete random variable on a sublinear expectation space $(\Omega,\mathcal{H},\mathbb{E})$. The corresponding uncertain probability distributions of $X$ are $\left\{p_{\theta }(X)\right\}_{\theta \in \Theta }$. The {\it nonlinear information entropy}, or abbreviated as {\it nonlinear entropy}, is formulated as  
	\begin{equation*}
	\hat{H}\left(X\right):=\sup_{\theta \in \Theta} \sum _{x}p_{\theta }\left(x\right)\log \frac{1}{p_{\theta }\left(x\right)}.
	\end{equation*}
\end{definition}

The amount of information embedded in uncertain-distribution information sources is described by the nonlinear information entropy presented above. Furthermore, to describe the uncertainty involved in information transmission processes, we define \textit{nonlinear joint entropy} and \textit{nonlinear conditional entropy} of two discrete random variables on a sublinear expectation space, as follows.

\begin{definition}
	Let $X,Y$ be two discrete random variables on a sublinear expectation space $(\Omega,\mathcal{H},\mathbb{E})$. The corresponding uncertain probability distributions of $X$ are $\{p_{\theta}(X)\}_{\theta \in \Theta}$, and the corresponding uncertain transition probability matrices of $Y$ given $X$ are $\{\boldsymbol{P}_{\lambda}(Y|X)\}_{\lambda \in \Lambda}$. Then
	\begin{enumerate}
		\item {the {\it nonlinear joint entropy} of $X$ and $Y$ is defined as:
			\begin{equation*}
			\hat{H}\left(X,Y\right):=\sup_{\theta \in \Theta} \sup_{\lambda \in \Lambda} \sum _{x,y}p_{\theta }\left(x\right)p_{\lambda }\left(y|x\right)\log \frac{1}{p_{\theta }\left(x\right)p_{\lambda }\left(y|x\right)},
			\end{equation*}}
		
		\item {the {\it nonlinear conditional entropy} of $Y$ given $X$ is defined as:
			\begin{equation*}
			\hat{H}\left(Y|X=x\right):=\sup_{\lambda \in \Lambda} \sum _{y}p_{\lambda }\left(y|x\right)\log \frac{1}{p_{\lambda }\left(y|x\right)}
			\end{equation*}
			and
			\begin{equation*}
			\hat{H}\left(Y|X\right):=\sup_{\theta \in \Theta} \sum _{x}p_{\theta }\left(x\right)\hat{H}\left(Y|X=x\right).
			\end{equation*}}
	\end{enumerate}
\end{definition}

The example below may help us develop an intuitive understanding of how nonlinear information entropy is affected by the uncertainty of the distributions. Let $p$ and $\epsilon$ be two constants in the interval $[0,1]$. Consider a discrete random variable $X$ with two possible values $\{0,1\}$. Suppose that the probability of the event $\{X=0\}$ is uncertain but takes value in $[\max(p-\epsilon , 0),\min(p+\epsilon,1)]$ and $\textrm{Pr}(X=1)=1-\textrm{Pr}(X=0)$. The uncertain probability distributions of $X$ are characterized as a family of probability distributions, which is denoted by $\{p_{q}(X)\}_{q\in [\max(p-\epsilon , 0),\min(p+\epsilon,1)]}$. More specifically, we have
\begin{equation*}
\{p_{q}(X)\}_{q\in [\max(p-\epsilon , 0),\min(p+\epsilon,1)]}=\{p_q=\{q,1-q\}|q\in [\max(p-\epsilon , 0),\min(p+\epsilon,1)] \}.
\end{equation*}
Note that the larger $\epsilon$, the greater the uncertainty of distributions of $X$.

\begin{figure}[htbp]
	\centering
	\includegraphics[scale=0.6]{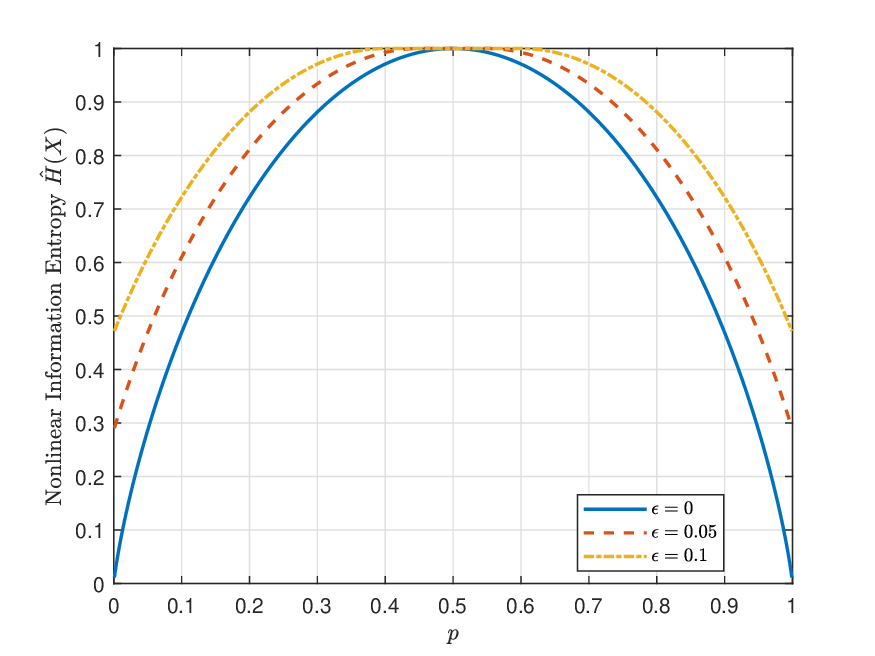}
	\captionsetup{font={scriptsize}}
	\caption{$\hat{H}(X)$ versus $p$, where $\epsilon=0$, $\epsilon=0.05$ and $\epsilon=0.1$.}
	\label{fig_2}
\end{figure}

A simple calculation tells us
\begin{equation*}
\hat{H}(X)=\underset{q\in [\max(p-\epsilon , 0),\min(p+\epsilon,1)]}{\sup } \left(q\log \frac{1}{q}+(1-q)\log \frac{1}{1-q}\right).
\end{equation*}
The relation between $\hat{H}(X)$ and $p$ is shown in Fig. \ref{fig_2}, where the red dashed line and the yellow dash-dotted line represent nonlinear information entropy $\hat{H}(X)$ with $\epsilon=0.05$ and $\epsilon=0.1$, respectively, while the blue solid line represents $H_{p}=p\log \frac{1}{p}+(1-p)\log \frac{1}{1-p}$ with the logarithmic base being 2. The value of the nonlinear information entropy $\hat{H}(X)$ increases with the value of $\epsilon$, which determines the degree of uncertainty of distributions of $X$.

For uncertain-distribution channel models, it is necessary to define new measures related to the signal transmission process $X \rightarrow Y$, where $X$ and $Y$ are two discrete random variables defined on a sublinear expectation space $(\Omega,\mathcal{H},\mathbb{E})$. In order to describe the uncertainty of the transmission process clearly, it is necessary to define \textit{nonlinear mutual information} under the framework of nonlinear information theory, in accordance with the concepts of nonlinear expectation and the uncertainty of distributions.

\begin{definition}
	Let $X,Y$ be two discrete random variables on a sublinear expectation space $(\Omega,\mathcal{H},\mathbb{E})$. The corresponding uncertain probability distributions of $X$ are $\{p_{\theta}(X)\}_{\theta \in \Theta}$, and the corresponding uncertain transition probability matrices of $Y$ given $X$ are $\{\boldsymbol{P}_{\lambda}(Y|X)\}_{\lambda \in \Lambda}$. Then the nonlinear mutual information between $X$ and $Y$ is defined as
	\begin{equation*}
	\overline{I}\left[X;Y\right]:=\sup_{\theta \in \Theta} \sup_{\lambda \in \Lambda } \sum _{x,y}p_{\theta }\left(x\right)p_{\lambda }\left(y|x\right)\log \frac{p_{\lambda }\left(y|x\right)}{\sum_{x}p_{\lambda }\left(y|x\right)p_{\theta }\left(x\right)}.
	\end{equation*}
	If there is also a discrete random variable $Z\sim \{p_{\theta}(Z)\}_{\theta \in \Theta'}$ such that the corresponding uncertain transition probability matrices of $X$ given $Z$ are $\{\boldsymbol{P}_{\lambda_1}(X|Z)\}_{\lambda_1 \in \Lambda_1}$ and the corresponding uncertain transition probability matrices of $Y$ given $X,Z$ are $\{\boldsymbol{P}_{\lambda_2}(Y|X,Z)\}_{\lambda_2 \in \Lambda_2}$, then the nonlinear conditional mutual information is defined as
	\begin{equation*}
	\overline{I}\left[X;Y|Z\right]:= \sup_{\theta \in \Theta'} \sup_{\substack{\lambda_1 \in \Lambda_1 \\ \lambda_2 \in \Lambda_2}} \sum _{x,y,z}p_{\theta}(z) p_{\lambda_1 }\left(x|z\right) p_{\lambda_2 }\left(y|x,z\right) \log \frac{p_{\lambda_2 }\left(y|x,z\right)}{\sum_{x}p_{\lambda_2 }\left(y|x,z\right)p_{\lambda_1 }\left(x|z\right)}.
	\end{equation*}
\end{definition}

\begin{remark}
	We also use the notation $\overline{I}\left[\{p_{\theta}(X)\}_{\theta \in \Theta};\{\boldsymbol{P}_{\lambda}(Y|X)\}_{\lambda \in \Lambda}\right]$ to denote nonlinear mutual information. This indicates that the uncertainty of distributions, rather than messages or random variables themselves, is the more significant factor within the framework of nonlinear information theory.
\end{remark}

Consider an uncertain-distribution channel with input alphabet $\mathcal{X}$ and output alphabet $\mathcal{Y}$. Due to the uncertainty of the actual transmission process, the corresponding uncertainty of the uncertain-distribution channel can be expressed as a family of transition probability matrices $\{\boldsymbol{P}_{\lambda } \in [0,1]^{\mathcal{X}\times \mathcal{Y}} \}_{\lambda \in \Lambda}$. We denote this uncertain-distribution channel by $[\mathcal{X},\{\boldsymbol{P}_{\lambda } \in [0,1]^{\mathcal{X}\times \mathcal{Y}}\}_{\lambda \in \Lambda },\mathcal{Y} ]$.

\begin{definition}
	Let $(\Omega,\mathcal{H},\mathbb{E})$ be a sublinear expectation space where $\mathbb{E}[\cdot]=\underset{\theta \in \Theta}{\sup} E_{P_{\theta}}[\cdot]$. An uncertain-distribution channel $[\mathcal{X},\{\boldsymbol{P}_{\lambda } \in [0,1]^{\mathcal{X}\times \mathcal{Y}}\}_{\lambda \in \Lambda },\mathcal{Y} ]$ is called memoryless if for any input codewords $x_1,\cdots,x_n$ and output codewords $y_1,\cdots,y_n$, the uncertain-distribution channel satisfies
	\begin{equation*}
	\{P_{\theta}(Y_1=y_1,\cdots,Y_n=y_n|X_1=x_1,\cdots,X_n=x_n)\}_{\theta \in \Theta}=\{\prod_{i=1}^{n}p_{\lambda}(y_i|x_i)\}_{\lambda \in \Lambda}.
	\end{equation*}
	where $X_1,\cdots,X_n$ denote the input random variables, $Y_1,\cdots,Y_n$ denote the output random variables.
\end{definition}

Without loss of generality, in this paper we only consider discrete memoryless uncertain-distribution channel, which is the most commonly used channel model.

To summarize, we newly define several fundamental concepts for nonlinear information theory, presenting a novel framework for understanding the uncertainties of information sources and communication channels by using nonlinear expectation theory. This study extends classical information theory to accommodate the complexities introduced by the uncertainty of probability models, offering new insights into the analysis of information and communication systems.

\section{Main Results of Nonlinear Information Theory}\label{sec4}
In this section we present the main results of our nonlinear information theory relying on a sublinear expectation space $(\Omega,\mathcal{H},\mathbb{E})$, where the uncertainty of the probability measures associated with $\mathbb{E}$ is characterized by $\mathcal{P}=\{P_{\theta}\}_{\theta \in \Theta}$, i.e.,
\begin{equation*}
\mathbb{E}[X]=\sup_{\theta \in \Theta}\int_{\Omega}XdP_{\theta }, \quad \forall X\in \mathcal{H}.
\end{equation*}

\subsection{Properties of New Measures in Nonlinear Information Theory}\label{sec4.1}

In this subsection, we prove some important properties for nonlinear information entropy, nonlinear joint entropy, nonlinear conditional entropy and nonlinear mutual information. These properties can be proved by taking appropriate supremum and infimum operations over the underlying set of probability measures, and the proofs are therefore omitted.

\begin{theorem}\label{theless}
	Let $X,Y$ be two discrete random variables on a sublinear expectation space $(\Omega,\mathcal{H},\mathbb{E})$. The corresponding uncertain probability distributions of $X$ are $\{p_{\theta}(X)\}_{\theta \in \Theta}$, and the corresponding uncertain transition probability matrices of $Y$ given $X$ are $\{\boldsymbol{P}_{\lambda}(Y|X)\}_{\lambda \in \Lambda}$. Then, we have
	\begin{equation*}
	\hat{H}\left(X,Y\right)\leq \hat{H}\left(X\right)+\hat{H}\left(Y|X\right).
	\end{equation*}
	The equality holds when the distributions of $X$ and transition probability matrices of $Y$ given $X$ become deterministic. Moreover, when $\hat{H}\left(Y|X\right)=0$, the above inequality degenerates into the equality $\hat{H}\left(X,Y\right)=\hat{H}\left(X\right)$. $\hfill\blacksquare$
\end{theorem}

According to this theorem, we can observe that in contrast to nonlinear joint entropy, nonlinear conditional entropy is increased more due to the introduction of the uncertainty of transition probability matrices. Note that the property presented by Theorem \ref{theless} is different from its counterpart in classical information theory, where only the equality holds.

\begin{theorem}\label{thegreat}
	Let $X,Y$ be two discrete random variables on a sublinear expectation space $(\Omega,\mathcal{H},\mathbb{E})$. The corresponding uncertain probability distributions of $X$ are $\{p_{\theta}(X)\}_{\theta \in \Theta}$, and the corresponding uncertain transition probability matrices of $Y$ given $X$ are $\{\boldsymbol{P}_{\lambda}(Y|X)\}_{\lambda \in \Lambda}$. Then, we have
	\begin{equation*}
	\overline{I}\left(X;Y\right)\geq \hat{H}\left(X\right)-\hat{H}(X|Y).
	\end{equation*}
	The equality holds when the distributions of $X$ and transition probability matrices of $Y$ given $X$ become deterministic. Moreover, when $\hat{H}\left(X|Y\right)=0$, the above inequality degenerates into the equality $\overline{I}\left(X;Y\right)=\hat{H}\left(X\right)$. $\hfill\blacksquare$
\end{theorem}

\begin{theorem}\label{thedata}
	Let $X,Y,Z$ be discrete random variables on a sublinear expectation space $(\Omega,\mathcal{H},\mathbb{E})$. Assume that for every $\theta\in \Theta$, under the probability measure $P_{\theta}$, $X\rightarrow Y\rightarrow Z$ forms a Markov chain. We have
	\begin{equation*}
	\overline{I}(X;Z)\leq \min(\overline{I}(X;Y),\overline{I}(Y;Z)).
	\end{equation*} $\hfill\blacksquare$
\end{theorem}

\begin{remark}
	Note that $\hat{H}\left(X|Y\right) \leq \hat{H}(X)$ does not hold in nonlinear information theory. However, in classical information theory, we have $H\left(X|Y\right) \leq H(X)$, where $H\left(X|Y\right)$ and $H(X)$ are the conditional entropy and the entropy, respectively.
\end{remark}

As the case extends to more than two random variables, the chain rule regarding nonlinear joint entropy and nonlinear conditional entropy exhibits an inequality form, and so does the chain rule of nonlinear mutual information. These insights are shown in Theorem \ref{thechain} and Theorem \ref{theless2}, respectively.

\begin{theorem}\label{thechain}
	For a discrete random variable sequence $X_{1},\cdots ,X_{n}$ defined on a sublinear expectation space $(\Omega,\mathcal{H},\mathbb{E})$, we have
	\begin{equation*}
	\hat{H}\left(X_{1},\cdots ,X_{n}\right)\leq \sum _{i=1}^{n}\hat{H}\left(X_{i}|X_{i-1},\cdots ,X_{1}\right).
	\end{equation*} $\hfill\blacksquare$
\end{theorem}

\begin{theorem}\label{theless2}
	For a discrete random variable sequence ${X_{1},\cdots,X_{n}}$ and a discrete random variable $Y$, both defined on a sublinear expectation space $(\Omega,\mathcal{H},\mathbb{E})$, we have
	\begin{equation*}
	\overline{I}\left(X_{1},X_{2},\cdots ,X_{n};Y\right)\leq \sum _{i=1}^{n}\overline{I}\left(X_{i};Y|X_{i-1},\cdots ,X_{1}\right).
	\end{equation*} $\hfill\blacksquare$
\end{theorem}

Let $X,\hat{X}$ be two discrete random variable defined on a sublinear expectation space $(\Omega,\mathcal{H},\mathbb{E})$. $\mathcal{X}$ denotes the set of the states of $X$. Consider the signal transmission process $X \rightarrow \hat{X}$, where $X$ and $\hat{X}$ represent input and output messages respectively, the probability of the communication error event $\{\hat{X}\neq X\}$ is also uncertain. That is, it obeys a family of probabilities $\{P_{\theta}(\hat{X}\neq X)\}_{\theta \in \Theta}$. For notational convenience, in what follows we shall simply denote $P_{\theta}(\hat{X}\neq X)$ by $P_{e,\theta}$ and denote the set $\{P_{\theta}(\hat{X}\neq X)\}_{\theta \in \Theta}$ by $\mathcal{P}_{e}$. The Fano inequality, which describes the relation between conditional entropy and error probability, is an important property in classical information theory. In the framework of nonlinear information theory, there is also a relation between nonlinear conditional entropy and a family of error probabilities, as stated below.

\begin{theorem}\label{thefano}
	Let $X,\hat{X}$ be two discrete random variables on a sublinear expectation space $(\Omega,\mathcal{H},\mathbb{E})$. $\mathcal{X}$ denotes the set of the states of $X$, $P_{e,\theta} = P_{\theta}(\hat{X}\neq X)$, and $\mathcal{P}_{e} = \{P_{\theta}(\hat{X}\neq X)\}_{\theta \in \Theta}$. Then we have
	\begin{equation*}
	\hat{H}\left(X|\hat{X}\right)\leq \sup_{P_{e,\theta}\in \mathcal{P}_{e}} \left[P_{e,\theta}\log \frac{1}{P_{e,\theta}}+\left(1-P_{e,\theta}\right)\log \frac{1}{1-P_{e,\theta}}\right]+\sup_{P_{e,\theta}\in \mathcal{P}_{e}} P_{e,\theta } \log \left(\| \mathcal{X} \| -1\right),
	\end{equation*}
	where $\parallel \mathcal{X} \parallel$ indicates the cardinality of the set $\mathcal{X}$. $\hfill\blacksquare$
\end{theorem}

\subsection{Source Coding in Nonlinear Information Theory}\label{sec4.2}

In the framework of nonlinear information theory, in order to reduce information redundancy and improve communication efficiency, it is necessary to encode and decode for the uncertain-distribution information sources similar to the case in classical information theory. Source coding in nonlinear information theory refers to the process of encoding and decoding for uncertain-distribution sources that is incorporated in a nonlinear communication model. The processes of encoding and decoding for uncertain-distribution sources are similar to those of source coding in classical information theory. However, it is important to note that information sources are defined on a sublinear expectation space $(\Omega,\mathcal{H},\mathbb{E})$, which implies that the sources have distribution uncertainty. Therefore, the probability models for analyzing the associated problems are uncertain as well, and the performance measure of the transmission process need to be redefined.

The encoding and decoding for source coding processes in nonlinear information theory can still be represented as functions. Specifically, a $(\| \mathcal{W} \|,n,f_n,g_n)$ nonlinear source code consists of:
\begin{enumerate}
	\item{ A source encoding function $f_n\colon \mathcal{X}^{n}\rightarrow \mathcal{W}$, assigns an index $W \in \mathcal{W}$ to each $\boldsymbol{X}^n\in \mathcal{X}^n$.}
	\item{ A source decoding function $g_n\colon \mathcal{W}\rightarrow \mathcal{X}^{n}$ that assigns an estimate $\hat{\boldsymbol{X}}^n\in \mathcal{X}^n$ to each index $W\in \mathcal{W}$.}
\end{enumerate}
For the above nonlinear source code, the source coding rate $R_{\textrm{s}}$ is defined as $\frac{\log \|\mathcal{W}\|}{n}$. After the uncertain-distribution information source generates source messages $\boldsymbol{X}^n=(X_1,\cdots,X_n)$, the source messages are encoded into an index $W\in \mathcal{W}$ by the encoding function $f_n \colon \mathcal{X}^{n}\rightarrow \mathcal{W}$. After determining the index $W'$ belonging to $\mathcal{W}$, $W'$ is processed by the decoding function $g_n \colon \mathcal{W}\rightarrow \mathcal{X}^{n}$, to estimate the original message as $\hat{\boldsymbol{X}}^{n}=(\hat{X}_1,\cdots,\hat{X}_n)=g(W')$. The performance of the above nonlinear source code is measured by the maximum and the minimum probabilities of the event that the estimate of the message is different from the message actually sent. Here, these probabilities are calculated as
\begin{equation*}
\mathbb{E}\left[I_{\{\hat{\boldsymbol{X}}^n\neq \boldsymbol{X}^n\}}\right]=\sup_{\theta \in \Theta } P_{\theta}(\hat{\boldsymbol{X}}^n\neq \boldsymbol{X}^n)=\sup_{\theta \in \Theta } P_{e,\theta}^{(n)},
\end{equation*}
\begin{equation*}
\mathcal{E}\left[I_{\{\hat{\boldsymbol{X}}^n\neq \boldsymbol{X}^n\}}\right]=\inf_{\theta \in \Theta } P_{\theta}(\hat{\boldsymbol{X}}^n\neq \boldsymbol{X}^n)=\inf_{\theta \in \Theta } P_{e,\theta}^{(n)}, 
\end{equation*}
where $\mathcal{E}$ represents the conjugate expectation of $\mathbb{E}$ and is defined by \eqref{eq:conjugate_expectation}. Intuitively, the above mentioned terms ``maximum probability'' and ``minimum probability'' correspond to the conservative and the aggressive strategies of designing nonlinear source code, respectively. Specifically, the maximum error probability criterion adopts a worst-case design viewpoint: it requires the error probability to be small uniformly over all distributions in the family, and is therefore suitable for applications where robustness against distribution uncertainty is critical. In contrast, the minimum error probability criterion adopts a best-case design viewpoint: it requires only that there exists some distribution in the family under which the error probability is small, and is therefore more optimistic. This criterion is suitable when the designer seeks to exploit the most favorable scenario permitted by the uncertainty family. Based on this insight, there are two criteria for evaluating the performance of nonlinear source code, namely the maximum error probability criterion and the minimum error probability criterion. The former refers to ensuring $\mathbb{E}\left[I_{\{\hat{\boldsymbol{X}}^n\neq \boldsymbol{X}^n\}}\right]$ to be sufficiently small, while the latter refers to ensuring $\mathcal{E}\left[I_{\{\hat{\boldsymbol{X}}^n\neq \boldsymbol{X}^n\}}\right]$ to be sufficiently small.

\begin{definition}\label{def:IID source}
	A sequence of messages denoted by $X_{1}, X_{2}, \cdots, X_{n},\cdots$ is said to be an \textit{IID source} if $X_{1}, X_{2}, \cdots, X_{n}, \cdots$ are IID.
\end{definition}

For a discrete IID source sequence $X_{1},X_{2}, \cdots,X_{n},\cdots$ on a sublinear expectation space, it is important to study the fundamental limit of its source coding rate. Note that $\hat{H}(X_1) = \hat{H}(X_2) = \cdots = \hat{H}(X_n) = \cdots$ holds true, since we consider an IID source sequence. We also consider the case where $V$ satisfies the following uniform boundedness condition: there exists a sequence $\{ M_n \}$ such that for every $n\geq1$, $\underset{\boldsymbol{x}^{n}\in \mathcal{X}^{n}}{\sum} V(\boldsymbol{x}^{n})\leq M_n$ and $\underset{n\rightarrow \infty}{\lim}\frac{\log M_n}{n}=0$.

\begin{theorem}\label{thesource}
	Let $X_{1},X_{2},\cdots,X_{n},\cdots$ be a discrete IID source sequence defined on a sublinear expectation space $(\Omega,\mathcal{H},\mathbb{E})$, the corresponding uncertain probability distributions of $X_1$ are \allowbreak $\left\{p_{\theta }(X_1)\right\}_{\theta \in \Theta }$. $(V,v)$ is the pair of capacities generated by $\mathbb{E}$. Suppose that $C_V(|\log \frac{1}{V(X_1)}|)<\infty$, $\{P_{\theta}\}_{\theta\in \Theta}$ is a countably-dimensional weakly compact family of probability measures on $(\Omega,\sigma(\mathcal{H}))$, and $V$ satisfies the uniform boundedness condition. Then, we have
	\begin{enumerate}
		\item {for any $R_{\textrm{s}} > \underset{\theta \in \Theta}{\inf} \underset{x}{\sum}p_{\theta }(x)\log \frac{1}{V(x)}$ and any positive value $\epsilon$, there is a sufficiently large $n$ and a $(||\mathcal{W}||,n,f_n,g_n)$ nonlinear source code with coding rate $R_{\textrm{s}}$ such that $\mathcal{E}\left[I_{\{\hat{\boldsymbol{X}}^n\neq \boldsymbol{X}^n\}}\right]< \epsilon$;}
		\item {for any $R_{\textrm{s}} > \hat{H}(X_{1})$ and any $\epsilon$ of positive value, there is a sufficiently large $n$ and a $(||\mathcal{W}||,n,f_n,g_n)$ nonlinear source code with coding rate $R_{\textrm{s}}$ such that $\mathbb{E}\left[I_{\{\hat{\boldsymbol{X}}^n\neq \boldsymbol{X}^n\}}\right]< \epsilon$.}
	\end{enumerate}
\end{theorem}

\textit{Proof}: See Appendix \ref{app3}. $\hfill\blacksquare$

\begin{remark}
	We emphasize that the conclusions in Theorem \ref{thesource} are achievability statements. In particular, the quantity $\underset{\theta \in \Theta}{\inf} \underset{x}{\sum}p_{\theta }(x)\log \frac{1}{V(x)}$ is also an achievability bound under the minimum error probability criterion. Furthermore, for each coding rate above the corresponding threshold, there exists a sufficiently large blocklength $n$ such that the required code exists. The theorem does not require that the code works for every sufficiently large $n$.
\end{remark}

We call Theorem \ref{thesource} the nonlinear source coding theorem. Note that the source coding rate limit in the first statement of Theorem \ref{thesource} and its classical counterpart yield different performance limits under the minimum error probability criterion.

\begin{remark}
	In classical information theory, for an information source $X$ with probability distribution $p(X)$, the limit of source coding rate is $\underset{x}{\sum}p(x)\log \frac{1}{p(x)}$, i.e., Shannon entropy. Furthermore, if the random variable is assumed to follow a specific probability distribution from a family of distributions $\{p_{\theta}(X)\}_{\theta\in \Theta}$ but the exact one is unknown, then the limit of source coding rate can be optimized to $\underset{\theta \in \Theta}{\inf} \underset{x}{\sum}p_{\theta }(x)\log \frac{1}{p_{\theta }(x)}$ under the minimum error probability criterion. By contrast, in the nonlinear information theory developed in this paper, we consider a more general case that the random variable is described by a family of distributions $\{p_{\theta}(X)\}_{\theta\in \Theta}$ but it does not follow any specific distribution within the family. Instead, the family of distributions is used as a conceptual tool to characterize the random variable's behavior without assuming it follows any particular distribution. This more general scenario and the relaxation of the traditional i.i.d. assumption make it possible to extend the traditional source coding rate limit. Because $V(x)$ is greater than $p_{\theta}(x)$ for any $\theta\in \Theta$, we know that the value of $\underset{\theta \in \Theta}{\inf} \underset{x}{\sum} p_{\theta }(x)\log \frac{1}{V(x)}$ is smaller than the value of $\underset{\theta \in \Theta}{\inf} \underset{x}{\sum} p_{\theta }(x)\log \frac{1}{p_{\theta }(x)}$.
\end{remark}

Theoretically, a more accurate representation of uncertain-distribution sources can be made by using Theorem \ref{thesource} compared with traditional representation of information sources. Although we have not formulated the precise asymptotic bounds of the source coding rate in nonlinear information theory, the methodology we have introduced to study nonlinear source coding is also valuable for studying the other parts of nonlinear information theory.

\subsection{Channel Coding in Nonlinear Information Theory}\label{sec4.3}
In this subsection, we first consider the process of channel coding under the nonlinear information theory framework as data passes through the uncertain-distribution channels, i.e. communication channels with uncertainty inherent in transition probability matrices. A $(M,n,\varphi_n,\psi_n)$ nonlinear channel code consists of:
\begin{enumerate}
	\item{ A message set $\mathcal{S}=\{1,2,\cdots,M\}$.}
	\item{ A channel encoding function $\varphi_n \colon \mathcal{S}\rightarrow \mathcal{X}^{n}$ that assigns a codeword $\boldsymbol{X}^n\in\mathcal{X}^n$ to each message $S\in \mathcal{S}$.}
	\item{ A channel decoding function $\psi_n \colon \mathcal{Y}^{n}\rightarrow \mathcal{S}$ that assigns an estimate $\hat{S}\in \mathcal{S}$ to each received sequence $\boldsymbol{Y}^n\in \mathcal{Y}^n$.}
\end{enumerate}
The channel coding rate of the $(M,n,\varphi_n,\psi_n)$ nonlinear channel code is defined as $\frac{\log M}{n}$. The transmission process using the uncertain-distribution channel coding scheme can be characterized as:
\begin{equation}\label{eq:trans_process}
S\xrightarrow{\varphi_n}\boldsymbol{X}^{n}\xrightarrow{\left\{\boldsymbol{P}_{\lambda } \in [0,1]^{\mathcal{X}\times \mathcal{Y}}\right\}_{\lambda \in \Lambda }}\boldsymbol{Y}^{n}\xrightarrow{\psi_n}\hat{S}.
\end{equation}
The error performance in the above message transmission process \eqref{eq:trans_process} is characterized as the maximum and the minimum probabilities of the event that the estimate of the message is different from the message actually sent. Here, these probabilities are calculated as
\begin{equation}\label{eq:maximum pro}
\mathbb{E}\left[I_{\{\hat{S}\neq S\}}\right]=\sup_{\theta \in \Theta} P_{\theta}\left(\hat{S}\neq S\right)=\sup_{\theta \in \Theta} P_{e, \theta}^{(n)},
\end{equation}
\begin{equation}\label{eq:minimum pro}
\mathcal{E}\left[I_{\{\hat{S}\neq S\}}\right]=\inf_{\theta \in \Theta} P_{\theta}\left(\hat{S}\neq S\right)=\inf_{\theta \in \Theta} P_{e, \theta}^{(n)}.
\end{equation}
Intuitively, the above mentioned terms ``maximum probability'' and ``minimum probability'' correspond to the conservative and the aggressive strategies of designing nonlinear channel code, respectively. Therefore, similar to the research of source coding in Section \ref{sec4.2}, there are also two criteria for evaluating the performance of nonlinear channel code, namely the maximum error probability criterion and the minimum error probability criterion. The former refers to ensuring $\mathbb{E}\left[I_{\{\hat{\boldsymbol{X}}^n\neq \boldsymbol{X}^n\}}\right]$ to be sufficiently small, while the latter refers to ensuring $\mathcal{E}\left[I_{\{\hat{\boldsymbol{X}}^n\neq \boldsymbol{X}^n\}}\right]$ to be sufficiently small. The purpose of channel coding in our nonlinear information theory is to make the maximum probability or the minimum probability of errors occurring as small as possible, while making the channel coding rate as large as possible.

\begin{remark}
	The decoded message $\hat{S}$ is related to the output sequence of the uncertain-distribution channel, which is characterized by $\left\{\boldsymbol{P}_{\lambda } \in [0,1]^{\mathcal{X}\times \mathcal{Y}}\right\}_{\lambda \in \Lambda }$. Therefore, the calculations of (\ref{eq:maximum pro}) and (\ref{eq:minimum pro}) are also affected by the family of transition probability matrices $\left\{\boldsymbol{P}_{\lambda } \in [0,1]^{\mathcal{X}\times \mathcal{Y}}\right\}_{\lambda \in \Lambda }$. However, in (\ref{eq:maximum pro}) and (\ref{eq:minimum pro}), only $\theta$ and $\Theta$ are present. This is because $\lambda$ and $\Lambda$ are implicitly but closely related to $\theta$ and $\Theta$. Essentially, all distribution families considered in this paper are derived from the family of probability measures $\mathcal{P}=\{P_{\theta}\}_{\theta \in \Theta}$, which is associated with sublinear expectation. For the sake of notational convenience, in this paper we generally use $\theta,\Theta$ to represent the distributions of random variables and $\lambda,\Lambda$ to represent the transition probability matrices.
\end{remark}

The number $R_{\textrm{c}}$ $\geq 0$ represents an achievable coding rate for the uncertain-distribution channel under the maximum error probability criterion or the minimum error probability criterion, if there exist a sequence of $(M, n, \varphi_n, \psi_n)$ nonlinear channel codes with coding rate $R_{\textrm{c}}$ such that $\underset{\theta \in \Theta}{\sup} P_{e,\theta}^{(n)}$ or $\underset{\theta \in \Theta}{\inf} P_{e,\theta}^{(n)} $ tends to $0$ as $n\rightarrow \infty$. For any uncertain-distribution channel $[\mathcal{X},\left\{\boldsymbol{P}_{\lambda } \in [0,1]^{\mathcal{X}\times \mathcal{Y}}\right\}_{\lambda \in \Lambda },\mathcal{Y} ]$, one of fundamental research problems is to determine the upper and the lower bounds of achievable channel coding rate under the maximum error probability criterion or the minimum error probability criterion. Let 
\begin{equation*}
\begin{split}
\overline{C}:=&\sup_{\{p_{\theta }\}\subset \{ p(x),x\in \mathcal{X} \} } \overline{I}\left[\{p_{\theta}\}_{\theta \in \Theta};\{\boldsymbol{P}_{\lambda } \in [0,1]^{\mathcal{X}\times \mathcal{Y}}\}_{\lambda \in \Lambda}\right]\\
=&\sup_{p\left(x\right)} \sup_{\lambda \in \Lambda} \sum _{x\in\mathcal{X},y\in\mathcal{Y}}p\left(x\right)p_{\lambda }\left(y|x\right)\log \frac{p_{\lambda}\left(y|x\right)}{\sum_{x\in \mathcal{X}}p\left(x\right)p_{\lambda}\left(y|x\right)}.
\end{split}
\end{equation*}
Then, we derive the upper bound of achievable coding rate for uncertain-distribution channels under the maximum error probability criterion in the following theorem.

\begin{theorem}\label{thechannel}
	Define a discrete memoryless uncertain-distribution channel $[\mathcal{X},\break \left\{\boldsymbol{P}_{\lambda } \in [0,1]^{\mathcal{X}\times \mathcal{Y}}\right\}_{\lambda \in \Lambda },\mathcal{Y} ]$ with a sublinear expectation space $(\Omega,\mathcal{H},\mathbb{E})$. For any sequence of $(M,n,\varphi_n,\psi_n)$ nonlinear channel codes with coding rate $R_\textrm{c}$, if the maximum probability of errors occurring, i.e.,  $\underset{\theta \in \Theta}{\sup } P_{e,\theta}^{(n)}$, tends to 0 as $n\rightarrow \infty $, then $R_{\textrm{c}}$ $\leq \overline{C}$ must hold.
\end{theorem}

\textit{Proof}: See Appendix \ref{app4}. $\hfill\blacksquare$

\begin{theorem}\label{thechannel1}
	Define a discrete memoryless uncertain-distribution channel $[\mathcal{X},\break \left\{\boldsymbol{P}_{\lambda } \in [0,1]^{\mathcal{X}\times \mathcal{Y}}\right\}_{\lambda \in \Lambda },\mathcal{Y} ]$ with a sublinear expectation space $(\Omega,\mathcal{H},\mathbb{E})$ where the uncertain probability measures associated with sublinear expectation $\mathbb{E}$ are $\{P_{\theta} \}_{\theta\in \Theta}$. $(V,v)$ is the pair of capacities generated by $\mathbb{E}$. Suppose that $\{P_{\theta}\}_{\theta\in \Theta}$ is a countably-dimensional weakly compact family of probability measures on $(\Omega,\sigma(\mathcal{H}))$. Then, for any $R_{\textrm{c}} < \overline{C}$ and any $\epsilon$ of positive value, there is a sufficiently large $n$ and a $(M,n,\varphi_n,\psi_n)$ nonlinear channel code with coding rate $R_{\textrm{c}}$ such that
	\begin{equation}\label{eq:channel theorem}
	\mathcal{E}\left[I_{\{\hat{S}\neq S\}}\right]=\inf_{\theta \in \Theta} P_{e,\theta}^{(n)} < \epsilon.
	\end{equation}
\end{theorem}

\textit{Proof}: See Appendix \ref{app4}. $\hfill\blacksquare$

\begin{remark}
	Theorem \ref{thechannel} provides a converse-type bound under the maximum error probability criterion. Theorem \ref{thechannel1} provides an achievability statement under the minimum error probability criterion, with the existence of a sufficiently large $n$ rather than for all sufficiently large $n$.
\end{remark}

Theorem \ref{thechannel} and Theorem \ref{thechannel1} are collectively called the nonlinear channel coding theorems. These theorems provide a theoretical foundation for better characterizing the performance of communications systems. Although we have not established the supremum of the achievable coding rate for uncertain-distribution channels, it offers a fresh perspective on exploring the performance of channel coding under uncertain-distribution channel models.

\subsection{Source Coding with Distortion in Nonlinear Information Theory}\label{sec4.4}
In this subsection we discuss source coding model with a non-negligible distortion under the nonlinear information theory framework.

We first consider the transmission process denoted by $X \rightarrow \hat{X}$. Suppose $X$ and $\hat{X}$ are two random variables defined on a sublinear expectation space $(\Omega,\mathcal{H},\mathbb{E})$. The corresponding uncertain probability distributions of $X$ are $\{p_{\theta}(X)\}_{\theta \in \Theta}$, and the corresponding uncertain transition probability matrices of $\hat{X}$ given $X$ are $\{\boldsymbol{Q}_{\lambda}(\hat{X}|X)\}_{\lambda \in \Lambda}$. For an uncertain-distribution source encoding process $f\colon \mathcal{X}^{n}\rightarrow \mathcal{W}$ and decoding process $g\colon \mathcal{W}\rightarrow \hat{\mathcal{X}}^{n}$, the source coding rate $R_{\textrm{s}}$ has already been defined in Section \ref{sec4.2} as $\frac{\log||\mathcal{W}||}{n}$. In reality, the information sources, especially the uncertain-distribution sources, cannot always communicate fully error-free, and a certain distortion typically exists. Therefore, it is necessary to specify a distortion measure $d(X,\hat{X})$ on $\mathcal{X} \times \hat{\mathcal{X}}$, and use sublinear expectation to describe the distortion.

\begin{definition}
	Let $X$ be a discrete random variable on a sublinear expectation space $(\Omega,\mathcal{H},\mathbb{E})$. The distortion measure is denoted as $d(X,\hat{X})$. For a transmission process $X\rightarrow \hat{X}$, the \textit{maximum expected distortion} and the \textit{minimum expected distortion} are respectively defined as:
	\begin{equation*}
	\mathbb{E}[d(X,\hat{X})]=\sup_{\theta \in \Theta}\sup_{\lambda \in \Lambda }\sum_{x,\hat{x}}p_{\theta}(x)q_{\lambda}(\hat{x}|x)d(x,\hat{x}),
	\end{equation*}
	\begin{equation*}
	\mathcal{E}[d(X,\hat{X})]=\inf_{\theta \in \Theta}\inf_{\lambda \in \Lambda }\sum_{x,\hat{x}}p_{\theta}(x)q_{\lambda}(\hat{x}|x)d(x,\hat{x}).
	\end{equation*}
\end{definition}

For notational convenience, we write $\mathbb{E}_{\boldsymbol{Q}}[d(X,\hat{X})]$ and $\mathcal{E}_{\boldsymbol{Q}}[d(X,\hat{X})]$ as the value of $\mathbb{E}[d(X,\hat{X})]$ and $\mathcal{E}[d(X,\hat{X})]$ when the transition probability matrices of $\hat{X}$ given $X$ become a deterministic transition probability matrix $\boldsymbol{Q}(\hat{X}|X)$. Then, the rate distortion function based on the nonlinear mutual information is defined as follows.

\begin{definition}
	Let $X$ be a discrete random variable on a sublinear expectation space $(\Omega,\mathcal{H},\mathbb{E})$ and the distortion measure is denoted as $d(X,\hat{X})$. The \textit{rate distortion function} based on the nonlinear mutual information is defined as:
	\begin{equation*}
	\hat{R}^{I}\left(D\right):=\inf_{\boldsymbol{Q}\left(\hat{X}|X\right)\colon \mathbb{E}_{\boldsymbol{Q}}\left[d\left(X,\hat{X}\right)\right]\leq D} \overline{I}\left[\{p_{\theta}(X)\}_{\theta \in \Theta};\boldsymbol{Q}(\hat{X}|X)\right].
	\end{equation*}
\end{definition}

It is easy to check that the properties described in Theorem \ref{thelim1} hold for the rate distortion function based on the nonlinear mutual information.
\begin{theorem}\label{thelim1}
	The rate distortion function based on the nonlinear mutual information, i.e., $\hat{R}^{I}\left(D\right)$, satisfies:
	\begin{enumerate}
		\item {$\hat{R}^{I}\left(D\right)$ decreases with respect to $D$.}
		\item {$\hat{R}^{I}\left(D\right)$ is a convex function on $[0,+\infty )$.}
	\end{enumerate}
\end{theorem}

\textit{Proof}: See Appendix \ref{app5}. $\hfill\blacksquare$

\begin{theorem}\label{thelim2}
	Let $X_{1}, X_{2}, \cdots, X_{n},\cdots$ be a discrete IID source sequence defined on a sublinear expectation space $(\Omega,\mathcal{H},\mathbb{E})$. The distortion measure is denoted as $d(X,\hat{X})$. $(V,v)$ is the pair of capacities generated by $\mathbb{E}$. Suppose that $\{P_{\theta}\}_{\theta\in \Theta}$ is a countably-dimensional weakly compact family of probability measures on $(\Omega,\sigma(\mathcal{H}))$. Then, for any $R_{\textrm{s}}>\hat{R}^{I}\left(D\right)$ and $\epsilon>0$, there must exist a sufficiently large $n$ and a $(||\mathcal{W}||,n,f_n,g_n)$ nonlinear source code, which has a coding rate $R_{\textrm{s}}$ and the corresponding transition probability matrix $\boldsymbol{Q}(\hat{X}|X)$, such that $\mathcal{E}_{\boldsymbol{Q}}[d(\boldsymbol{X}^{n},\hat{\boldsymbol{X}}^{n})]\leq D+\epsilon$.
\end{theorem}

\textit{Proof}: See Appendix \ref{app5}. $\hfill\blacksquare$

Theorem \ref{thelim2} is called the nonlinear rate-distortion source coding theorem and is an achievability statement under the minimum expected distortion criterion. In practice, due to the uncertainty of distributions, it is challenging to ensure error-free transmission of messages. The classical rate-distortion source coding theorem provides the limit of source coding rate to control distortion when the distributions are deterministic. As far as the uncertainty of distributions is concerned, the proposed nonlinear rate-distortion source coding theorem is more suitable in real world circumstances.

\section{Application Examples}\label{sec5}
\subsection{Bernoulli Type Uncertain-Distribution Information Sources}\label{sec5.1}
In this subsection, an example is given for uncertain-distribution information sources by analogy with the Bernoulli type experiments having ambiguity, which demonstrates the importance of considering the uncertainty of probability models.

Consider binary message sequences composed of $0$ and $1$. Under the classical information theory, a transmitter sends $0$ or $1$ each time and the individual probabilities of sending $0$ and $1$ remain unchanged throughout the sending process. However, under the nonlinear information theory, the transmitter does not have fixed probabilities of sending $0$ and $1$. For example, let us consider a random variable $X$ with input alphabet $\mathcal{X}=\{0,1\}$. Suppose that the probability of the event $\{X=0\}$ is uncertain and it takes value in the interval $[\frac{1}{3},\frac{1}{2}]$. Therefore, the probability of the event $\{X=1\}$ is also uncertain and it takes value in the interval $[\frac{1}{2},\frac{2}{3}]$. Then, the uncertainty  probability distributions of $X$ are expressed as 
\begin{equation*}
\{p_q(X)\}_{q\in [\frac{1}{3},\frac{1}{2}]}:=\{p_q=\{q,1-q\}|q \in [\frac{1}{3},\frac{1}{2}] \}.
\end{equation*}

Consider the message sequence of length $N$ denoted by $X_1,\cdots,X_N$, where $X_i$ and $X$ are identically distributed for $i=1,\cdots,N$. In the following simulation, we construct a synthetic dataset\footnote{Throughout this paper, all uncertain-distribution datasets are generated using the method described in this example.} to illustrate the scenario of distribution uncertainty. The generation procedure is as follows: we randomly select $K=10$ distinct distributions $\{ p_{q_1},\cdots, p_{q_K} \}$ from the family $\{p_q(X)\}_{q\in [\frac{1}{3},\frac{1}{2}]}$, and for each selected distribution $p_{q_k}$, we generate $N/K$ samples. The resulting $K$ blocks are then concatenated and randomly permuted to eliminate any block order effect. The use of a small fixed number $K$ is necessary. If, instead, a new distribution were drawn from the family independently at each index, then by the law of large numbers the overall dataset would effectively obey a single mixed distribution, which would bring us back to the classical setting where a single underlying distribution exists, a setting that our theoretical framework rules out. The adopted blockwise procedure avoids this issue and ensures that the resulting dataset does not follow any single distribution in the family, thereby providing a meaningful illustration of the non-classical setting.

For the above message sequence, the traditional probability theory and methods do not work effectively no matter which deterministic Bernoulli distribution is used for approximating. The confidence level of fitting any deterministic Bernoulli distribution to the samples of the above message sequence of length $1024$ are shown in Fig. \ref{fig_3}. The abscissas $p$ mean that the Bernoulli distribution $\mathcal{B}(1,p)$ is used for approximating the samples of the message sequence. The ordinate represents the confidence level of this approximation. We can see that even the best fitting result can only achieve a confidence level of about $95\%$, which may be regarded unreliable in situations that require very high accuracy, such as the application scenarios of autonomous driving and precision manufacturing.

\begin{figure}[htbp]
	\centering
	\includegraphics[scale=0.6]{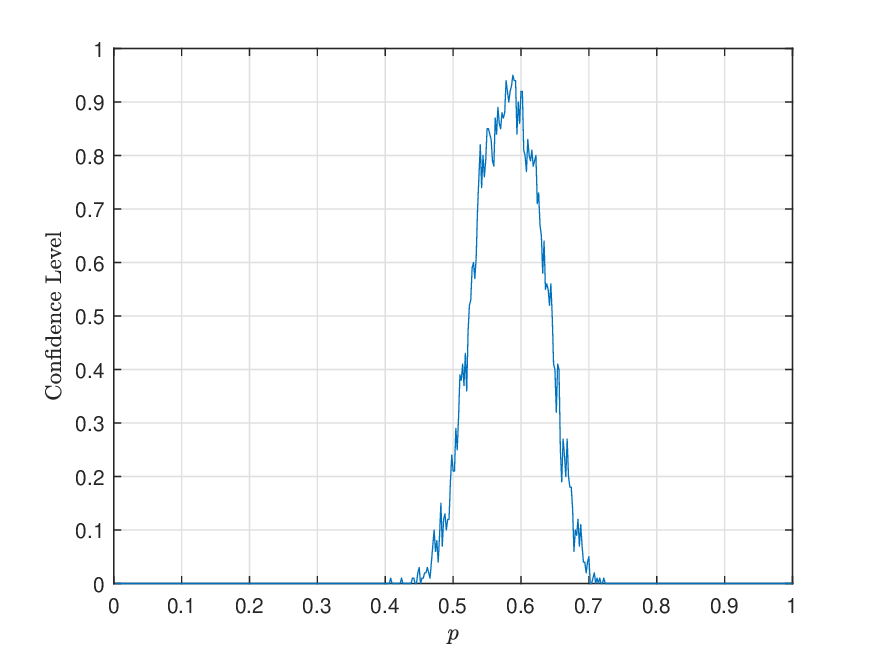}
	\captionsetup{font={scriptsize}}
	\caption{The confidence level of fitting deterministic Bernoulli distributions to the samples of a length-$1024$ message sequence with uncertain distributions $\{p_q(X)\}_{q\in [\frac{1}{3},\frac{1}{2}]}$.}
	\label{fig_3}
\end{figure}

\begin{remark}
	The simulation above is a heuristic demonstration. More precisely, within the nonlinear expectation theory, there is no canonical notion of "sampling from an uncertain distribution," since the framework does not assume the existence of any underlying true distribution. In practice, given a set of actual data, the nonlinear expectation framework is then employed to assess and quantify the distribution uncertainty inherent in the data. In a numerical simulation, however, no such actual data are available, so it is necessary to construct a synthetic instance that exhibits distribution uncertainty. The blockwise generation procedure described above provides a straightforward and intuitive way to achieve this for illustrative purposes. The simulation is intended to provide a concrete visualization of the concepts discussed in this paper. If desired, one may also start with a set of real-world data and carry out a similar simulation study using the same framework.
\end{remark}

Consider the above random variable $X$ and the corresponding probability distributions family $\{p_q(X)\}_{q\in [\frac{1}{3},\frac{1}{2}]}$. The sublinear expectation of $X$ is $\mathbb{E}[X]:=\underset{p_q\in \{p_q(X)\}_{q\in [\frac{1}{3},\frac{1}{2}]}}{\sup}E_{p_q}[X]$, and the pair of capacities $(V,v)$ is defined as:
\begin{equation*}
V(A):=\sup_{p_q\in \{p_q(X)\}_{q\in [\frac{1}{3},\frac{1}{2}]}}p_q(A), 
\end{equation*}
\begin{equation*}
\quad v(A):=\inf_{p_q\in \{p_q(X)\}_{q\in [\frac{1}{3},\frac{1}{2}]}}p_q(A).
\end{equation*}
It is easy to verify that $\mathbb{E}[X]=\frac{2}{3}$, $\hat{H}(X)=1$ and for any $p_q\in \{p_q(X)\}_{q\in [\frac{1}{3},\frac{1}{2}]}$
\begin{equation*}
\overline{\mu} = \mathbb{E}\left[\log \frac{1}{p_q(X)}\right] = \frac{1}{2}\left[\log \frac{1}{q}+\log \frac{1}{1-q}\right],
\end{equation*}
\begin{equation*}
\underline{\mu} = -\mathbb{E}\left[-\log \frac{1}{p_q(X)}\right] = \frac{1}{3}\log \frac{1}{q}+\frac{2}{3}\log \frac{1}{1-q}.
\end{equation*}

More generally, consider an uncertain-distribution information source $X$ with two input values $\{0,1\}$. Suppose that the probability of the event $\{X=0\}$ is uncertain and takes value in the interval $\Theta=[\max(p-\epsilon , 0),\min(p+\epsilon,1)]$ and $\textrm{Pr}(X=1)=1-\textrm{Pr}(X=0)$. Then the uncertain probability distributions of $X$ are represented by a family of probability distributions $\{p_{q}(X)\}_{q\in \Theta}$, which expressed as
\begin{equation*}
\{p_{q}(X)\}_{q\in \Theta}=\{p_q=\{q,1-q\}|q\in \Theta \}.
\end{equation*}

For the uncertain-distribution information source $X$, we visualize the results of Theorem \ref{thesource} in Fig. \ref{fig_4} to illustrate how the two performance limits (i.e., nonlinear information entropy and the achievability bound under the minimum error probability criterion) vary with the parameter $p$ for different levels of distribution uncertainty $\epsilon$, in comparison with the classical Shannon entropy. The dash-dotted lines represent the cluster point of the source coding rate of the uncertain-distribution source with $\epsilon=0.02$ and $\epsilon=0.03$ under the minimum error probability criterion, and are calculated by $\underset{\theta \in \Theta}{\inf} \sum _{x}p_{\theta }(x)\log \frac{1}{V(x)}$. The dashed lines represent the nonlinear information entropy of the uncertain-distribution sources with $\epsilon=0.02$ and $\epsilon=0.03$, and are calculated by $\underset{\theta \in \Theta }{\sup } \sum _{x}p_{\theta }(x)\log \frac{1}{p_{\theta }(x)}$. For the convenience of comparison, we also draw the classical Shannon entropy of the information source without distribution uncertainty, which is corresponding to the case $\epsilon=0$ and represented by the solid line in Fig. \ref{fig_4}. We observe that the cluster point of the source coding rate of the uncertain-distribution source that we found is notably smaller than the classical Shannon entropy of the information source without distribution uncertainty. This is a consequence of the more optimistic nature of the minimum error probability criterion, which allows the code to perform well under the most favorable distributions within the uncertainty family, rather than under a fixed single distribution. Thus, the two limits correspond to different performance criteria and should be understood as distinct performance bounds.

\begin{figure}[htbp]
	\centering
	\includegraphics[scale=0.6]{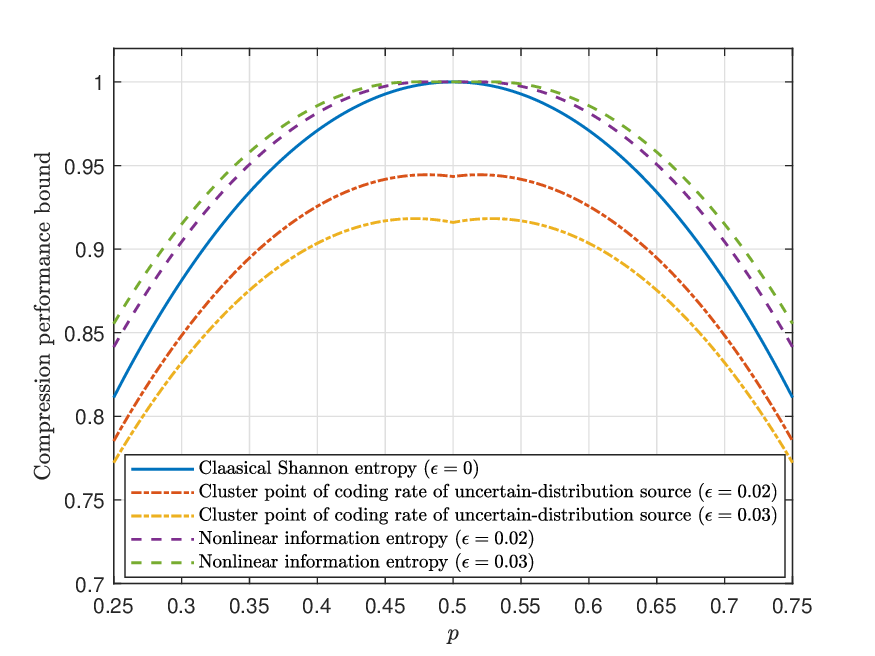}
	\captionsetup{font={scriptsize}}
	\caption{Comparison of the compression performance bound of an uncertain-distribution source and of the corresponding source without distribution uncertainty.}
	\label{fig_4}
\end{figure}

\subsection{Uncertain-Distribution Binary Symmetric Channels}\label{sec5.2}
In this subsection, we provide an example of uncertain-distribution channel, in order to illustrate how the family of transition probability matrices of an uncertain-distribution channel arises and visualize the result of Theorem \ref{thechannel1}. In the classical information theory, a typical and commonly used discrete channel model is the binary symmetric channel (BSC), and its transition probability matrix is represented as
\begin{equation}\label{BSC}
Q=\left[
\begin{matrix}
1-p & p \\
p & 1-p
\end{matrix}
\right],
\end{equation}
where $p$ denotes the probability of state transition. This model assumes that the channel introduces errors with a fixed probability $p$, which is independent of the input signal. However, in practical communication systems, the probability models may not always be precisely known due to various factors such as noise, interference, and imperfections in the channels or systems. Therefore, it is necessary to consider the uncertainty of probability models.

For a BSC characterized by (\ref{BSC}), an example\footnote{This is one possible realization of a BSC.} can be given as
\begin{equation}\label{eq:BSC}
Y=X \oplus Z,
\end{equation}
where $X$ is the input signal, $Y$ is the output signal, and $Z$ is the noise, all taking binary values from $\{0,1\}$. The noise $Z$ is modeled as a Bernoulli random variable with parameter $p$, meaning that $\textrm{Pr}(Z=1)=p$ and $\textrm{Pr}(Z=0)=1-p$. This model implies that the output $Y$ is the modulo-2 sum of the input $X$ and the noise $Z$. For instance, if $X=0$ and $Z=1$, then $Y=1$; if $X=1$ and $Z=1$, then $Y=0$. This representation highlights that the distribution of noise $Z$ is precisely known.

However, as we mentioned in Section \ref{sec1} about the research motivation of nonlinear information theory, the exact distribution of $Z$ may not exist in practical applications. In this case, we can assume that the probability of $Z=1$ lies within an interval $\Theta=[\max(p-\epsilon , 0),\min(p+\epsilon,1)]$, where $\epsilon$ represents the degree of uncertainty of the distribution of $Z$, and this uncertainty naturally leads to an uncertainty of the transition probability matrix of the channel. Given the uncertainty of the distributions of $Z$, the transition probability matrix of the channel model $Y = X \oplus Z$ is no longer fixed but varies within a set of possible matrices. Define $1-\Theta$ as the set $\{1-q|q\in \Theta\}$, then the uncertainty of transition probability matrix is expressed as
\begin{equation}\label{nonlinear BSC}
\left\{ \left[
\begin{matrix}
1-q & q \\
q & 1-q
\end{matrix}
\right] \Bigg| q\in \Theta \right\},
\end{equation}
and the corresponding uncertain-distribution BSC is shown in Fig. \ref{fig_5}.

\begin{figure}[htbp]
	\centering
	\includegraphics[scale=0.8]{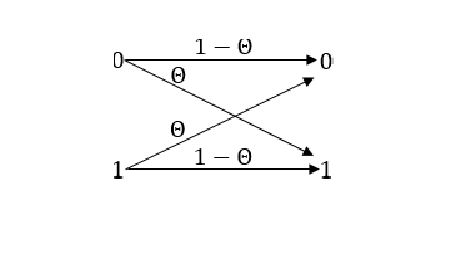}
	\captionsetup{font={scriptsize}}
	\caption{An example of uncertain-distribution BSC, where $\Theta$ is the interval $[\max(p-\epsilon , 0),\min(p+\epsilon,1)]$.}
	\label{fig_5}
\end{figure}

For an uncertain-distribution channel, to obtain the family of transition probability matrices, the $\varphi$-max-mean algorithm \cite{peng2017theory} within the nonlinear expectation theory can be employed to estimate the required parameters. We still take the aforementioned uncertain-distribution BSC as an example. Then the parameters we need to estimate for obtaining the family of transition probability matrices are $p$ and $\epsilon$, which can be estimated based on
\begin{equation*}
\min(p+\epsilon,1) = \mathbb{E}[Z], \quad \max(p-\epsilon , 0)=-\mathbb{E}[-Z].
\end{equation*}
More specifically, upon obtaining a finite sample set $\{z_i\}_{i=1}^{n\times m}$ of random variable $Z$, we can calculate
\begin{equation*}
\overline{\mathbb{E}}[Z]=\max_{1\leq i \leq m}\frac{1}{n}\sum_{j=1}^{n}z_{n(i-1)+j}
\end{equation*}
and 
\begin{equation*}
\underline{\mathbb{E}}[Z]=\min_{1\leq i \leq m}\frac{1}{n}\sum_{j=1}^{n}z_{n(i-1)+j}
\end{equation*}
to get the estimates of $\mathbb{E}[Z]$ and $-\mathbb{E}[-Z]$. According to \cite{jin2021optimal}, these estimates are also unbiased. Therefore, similar to classical information theory, we can still use a relatively short training sequence to learn the parameters of a given channel model.

\begin{figure}[htbp]
	\centering
	\includegraphics[scale=0.6]{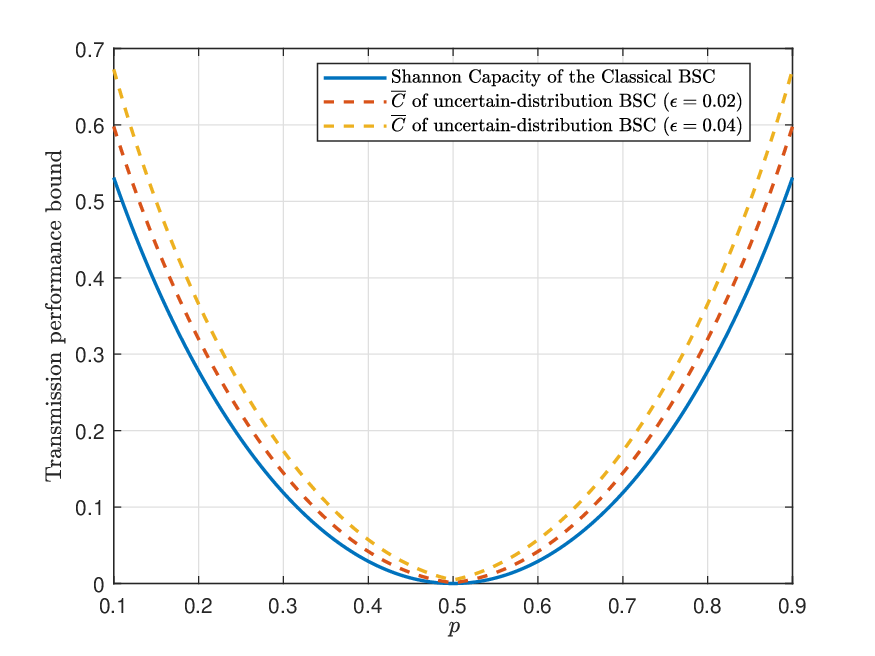}
	\captionsetup{font={scriptsize}}
	\caption{Comparison of the transmission performance bound of uncertain-distribution BSC and classical BSC.}
	\label{fig_6}
\end{figure}

To visualize the difference between the upper limit under the minimum error probability criterion and the classical Shannon capacity, we plot these quantities together in Fig. \ref{fig_6}. The upper limit of the achievable coding rate of the uncertain-distribution BSC that is characterized by (\ref{nonlinear BSC}) with $\epsilon=0.02$ and $0.04$ under the minimum error probability criterion is shown by the dashed lines. Note that the upper limit is denoted as $\overline{C}$, which is given by jointly considering Theorem \ref{thechannel} and Theorem \ref{thechannel1}. The Shannon capacity of the classical BSC that is characterized by (\ref{BSC}) is plotted as the solid line. We observe that under the given degrees of distribution uncertainty, the upper limit $\overline{C}$ is substantially larger than the Shannon capacity of the classical BSC, which assumes no distribution uncertainty. This is a consequence of the different criteria adopted; the two limits are obtained under different performance criteria and should therefore be understood as distinct performance bounds.

\subsection{Distinction Between Uncertain-Distribution Channel and A Class of Channels}\label{sec5.3}

In the context of classical information theory, the concept of ``a class of channels''\footnote{Both compound channels and arbitrarily varying channels fall within the broader framework of communication over a class of channels. Here, we use the term ``a class of channels'' as a generic reference.} has been proposed and researched (for example in \cite{blackwell1959capacity,blackwell1960capacities,ahlswede1969correlated,ahlswede1970note,ahlswede1978elimination,ahlswede1980method} and so on). In this paper, we propose a concept ``uncertain-distribution channel''. Both concepts utilize a family of transition probability matrices to characterize the channel behavior, but they differ fundamentally in how these distributions are applied and interpreted. In this subsection, we aim to clarify the distinctions between the two concepts through theoretical definitions and an illustrative example.

A class of channels is described by a family of transition probability matrices. However, the behavior of a class of channels conforms to one of the transition probability matrices within this family at any given observation slot (e.g., time slot). This implies a level of resolvability where the channel's stochastic behavior can be pinpointed to a specific transition probability matrix within the family under certain conditions. As a toy example, consider the simple communication channel $Y=X\oplus Z$, where $X$ is the input signal taking values in $\{0,1\}$, and $Z$ is a noise random variable also taking values in $\{0,1\}$. Then, the family of transition probability matrices is derived from specific observations of $Z$. Suppose we observe the value of $Z$ at a given time slot. If $Z=0$, the channel behavior corresponds to one specific transition probability matrix $\begin{pmatrix}
1 & 0 \\
0 & 1
\end{pmatrix}$, and if $Z=1$, it corresponds to $\begin{pmatrix}
0 & 1 \\
1 & 0
\end{pmatrix}$. This channel can also be viewed as a channel model for a jamming scenario in which the jammer chooses the states.

An uncertain-distribution channel is also characterized by a family of transition probability matrices. However, similar to the explanation presented in Appendix \ref{app1}, the actual behavior of the channel does not conform to any single transition probability matrix within this family. Instead, it represents a general concept where the channel's behavior is inherently uncertain and cannot be precisely captured by a specific transition probability matrix. For example, we still consider $Y=X\oplus Z$, and in this scenario, the distribution of $Z$ is uncertain and cannot be precisely specified. Based on the nonlinear expectation theory, the corresponding uncertain probability distributions of $Z$ are characterized as a family of probability distributions. Therefore, the transition probability matrix of the channel $Y=X\oplus Z$ is also uncertain. This uncertainty reflects the inherent ambiguity in the stochastic behavior of the channel.

To further elucidate the distinction from a physical perspective, it is worth examining the nature of uncertainty underlying each model. In classical models, such as compound channels and arbitrarily varying channels, the uncertainty is aleatory in the sense that the channel is assumed to be governed by a specific member of the family at each transmission, even if that member is unknown to the observer or chosen adversarially. The existence of a well-defined probabilistic law is always presupposed; the uncertainty lies in the observer's lack of knowledge about which law governs the transmission. In contrast, the uncertain-distribution channel proposed in this paper takes the entire family of transition matrices as the fundamental object of study. The framework does not presuppose that any particular member of the family is the "true" transition matrix. Instead, the family itself characterizes the channel's behavior holistically. This is conceptually similar to Shannon's treatment of uncertainty, where the probability distribution of a random variable is studied directly, without reducing the analysis to what would happen if the random variable took a specific value. In a similar spirit, our study of a random variable in the sublinear expectation space is naturally formulated in terms of the family of distributions, without reducing it to the individual transition matrices it contains. At this level, while both models employ a family of transition matrices, the underlying physical assumptions and the resulting information-theoretic objectives are fundamentally different.

Indeed, in classical radio propagation theory, the propagation environment at any specific moment is usually regarded as deterministic at the physical level; randomness is typically introduced through noise. The uncertain-distribution model does not dispute the physical fact that a single underlying state exists at any given time, even when that state is unknown. Rather, it addresses a different level of abstraction: it concerns the modeling of the channel, not the physical reality of the channel. In many practical scenarios, particularly in high-mobility, cross-medium, or heterogeneous environments, the channel is too complex to be accurately captured by a finite set of well-defined states. The family of transition matrices is not meant to suggest that multiple matrices coexist at a single instant, but rather to capture the model uncertainty arising from incomplete information about the channel's statistical behavior. The uncertain-distribution framework offers a systematic way to quantify and cope with model uncertainty in channel modeling, complementing rather than contradicting the classical engineering intuition.

In essence, the distinction between the two concepts mainly lies in the resolvability of their stochastic behaviors. ``Uncertain-distribution channel'' embodies an inherent and persistent uncertainty, whereas ``a class of channels'' implies conditional determinability. This distinction is critical in modeling communication systems, as it reflects the fundamental differences in how uncertainty is handled in channels that are randomly dynamic across a wide range of physical propagation environments, such as in high mobility cross-medium communications encountered in space-air-ground-sea integrated wireless networks \cite{10679173,11006056,dong2025uplink}.

\section{Conclusion}\label{sec6}
In this paper, we have established a nonlinear information theory within the framework of nonlinear expectation theory. We study fundamental problems for uncertain-distribution information sources and uncertain-distribution channels defined on sublinear expectation spaces. The concepts of nonlinear information entropy, nonlinear joint entropy, nonlinear conditional entropy and nonlinear mutual information are introduced, and several important properties such as the chain rule and the Fano inequality are generalized. Based on the strong law of large numbers under sublinear expectations, we establish the nonlinear source coding theorem, nonlinear channel coding theorem and nonlinear rate-distortion source coding theorem. These results characterize fundamental performance limits for uncertain-distribution sources and channels under different criteria. The proposed framework constitutes a generalization from classical Shannon information theory based on deterministic probability models to nonlinear information theory based on uncertain probability models. When the distributions of source messages and the transition probability matrices of communication channels become deterministic, the results obtained here reduce to their classical counterparts. This work represents a potential paradigm shift in information theory, and provides a theoretical foundation for further investigations of information-theoretic problems under distributional uncertainty.

{\appendices

\section{Why Use Nonlinear Expectation Theory--An Intuitive Explanation}\label{app1}

In this appendix, we explain the necessity of employing nonlinear expectation theory to characterize the uncertainty of probability distributions. To this end, we must first distinguish between two cases that characterize random variables with a distribution family. These cases reflect different levels of knowledge and assumptions about the underlying probability distributions, which influence the applicability and effectiveness of traditional probabilistic methods. In the first case, if the classical probability theory framework is used, a random variable is always assumed to follow a single specific probability distribution. However, if the single specific distribution is unknown, the random variable is usually assumed to follow a distribution selected from a family of candidate distributions, with the selection depending on the value of specific control variables. In the second case, i.e., the case that employs the nonlinear expectation theory, we assume that the random variable does not follow any specific distribution within a family of distributions. Instead, the entire family of distributions is used in a general manner to describe the random variable. These two cases are fundamentally different. Existing research works based on the classical probability theory correspond to the first case, where the probability model (i.e., the probability space) is deterministic. However, in many applications the probability model itself can also be uncertain, which leads to the emergence of the second case.

In the first case, we acknowledge the existence of a particular distribution characterizing the random variable, but the specific form or parameters of this distribution are uncertain. Consequently, the random variable is described by assuming that there exists a family of candidate distributions, to account for all possible behaviors. For example, consider a scenario where the random variable is believed to follow a normal distribution, but the precise values of the mean and variance are unknown. Then, we can consider a family of normal distributions with varying means and variances under different observation slots, and each observation slot corresponds to a single specific normal distribution. This approach allows us to capture the range of possible behaviors of the random variable while acknowledging the uncertainty of the distribution's parameters to some degree. However, it is worth noting that at each observation slot the uncertain distribution is degraded to a deterministic distribution. Traditional probabilistic methods, such as maximum likelihood estimation and Bayesian inference, are commonly employed to infer the unknown parameters from observed data. These methods assume that the underlying probability model (i.e., the probability space) is deterministic but unknown, and they aim to identify the most likely distribution that fits the observed data, within the family.

To elaborate a little further, let us consider a practical application involving a dataset $\{x_1,\cdots,x_n\}$. Even when we do not know the exact distribution that this dataset follows, we can still rely on the law of large numbers in probability theory to estimate the expectation of the underlying distribution. Specifically, we compute the sample mean $\sum_{i=1}^{n}x_i$ as an estimate of the distribution's expectation. This common practice is based on the implicit assumption that the underlying probability model is deterministic. In other words, we assume that the data are generated from a single, albeit unknown, distribution, and the sample mean provides a consistent estimate of the true expectation as the sample size grows. This approach aligns well with traditional probabilistic methods, which assume a deterministic but unknown probability model and aim to infer its parameters from observed data.

In contrast, the second case represents a more general and abstract framework. Here, the random variable does not follow any specific distribution within the family. Instead, the family of distributions is used as a conceptual tool to describe the random variable's behavior without assuming that it follows any specific distribution. This approach is particularly useful when dealing with complex systems or when the exact nature of the random variable is not well understood. For instance, in highly complex or randomly dynamic systems, such as those encountered in financial markets, biological networks, communication networks, or quantum information processing, the behavior of random variables may be influenced by numerous factors that are difficult to model explicitly. In these situations, using a family of distributions to describe the unpredictable nature of random variables provides a more general framework. It allows for a broader range of possibilities and acknowledges the inherent uncertainty of the system model, which may not be captured by any single distribution within the family.

To further demonstrate the uncertainty of distribution of a random variable, one can randomly select a set of actual data samples to perform the following calculations. Here, we randomly choose 10000 samples, denoted by $\{x_i\}_{i=1}^{10000}$, from a random variable $X$ whose probability distribution is uncertain and characterized by $\{p_q(X)\}_{q\in [\frac{1}{3},\frac{1}{2}]}:=\{p_q=\{q,1-q\}|q \in [\frac{1}{3},\frac{1}{2}] \}$.\footnote{This random variable is also described in Section \ref{sec5.1}. The data are generated using the same method as in that section.} For $n=5000,5001,\cdots,10000$, we calculate
\begin{equation}\label{eq:upper_means}
\overline{\mathbb{E}}_n:=\max_{1\leq i \leq 10}\frac{1}{500}\sum_{j=1}^{500}x_{n-5000+(i-1)\times 500+j}
\end{equation}
and 
\begin{equation}\label{eq:lower_means}
\underline{\mathbb{E}}_n:=\min_{1\leq i \leq 10}\frac{1}{500}\sum_{j=1}^{500}x_{n-5000+(i-1)\times 500+j}. 
\end{equation}
The statistics $\overline{\mathbb{E}}_n$ and $\underline{\mathbb{E}}_n$ reflect the upper means and the lower means of samples, respectively. The values of $\overline{\mathbb{E}}_n$ and $\underline{\mathbb{E}}_n$ are shown in Fig. \ref{fig_7}.

\begin{figure}[h]
	\centering
	\includegraphics[scale=0.6]{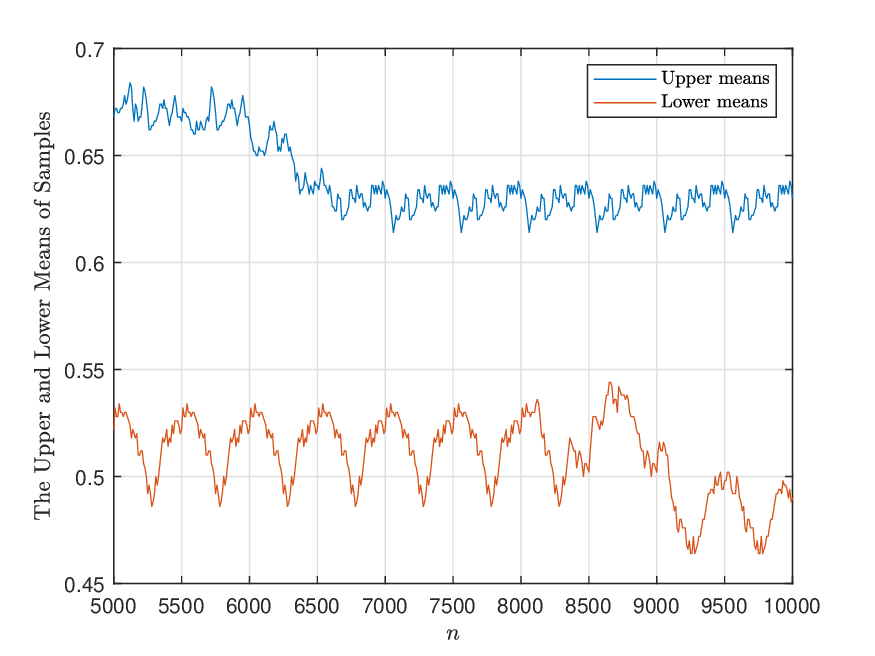}
	\captionsetup{font={scriptsize}}
	\caption{The upper means and lower means of samples of the random variable $X$ whose distribution is uncertain.}
	\label{fig_7}
\end{figure}

The statistic $\overline{\mathbb{E}}_n$ represents the maximum average value of 500 consecutive samples within a specific sliding window of 5000 samples (from $x_{n-5000}$ to $x_n$). This window is divided into 10 non-overlapping sub-windows, each containing 500 samples. Eq. \eqref{eq:upper_means} calculates the average of each sub-window and then takes the maximum of these averages.

The statistic $\underline{\mathbb{E}}_n$ represents the minimum average value of 500 consecutive samples within the same sliding window of 5000 samples (from $x_{n-5000}$ to $x_n$). Similar to $\overline{\mathbb{E}}_n$, the window is divided into 10 sub-windows, and Eq. \eqref{eq:lower_means} calculates the average of each sub-window and then takes the minimum of these averages.

We can observe in Fig. \ref{fig_7} that there is a significant gap\footnote{The gap remains visible even with an increased sample size. We use a sample size of 10000 merely as an example.} between $\overline{\mathbb{E}}_n$ and $\underline{\mathbb{E}}_n$. The fact that the maximum and minimum averages within this window differ substantially suggests that the behavior of $X$ is not consistent across different sub-windows. By contrast, let us consider a sequence of i.i.d. (independent and identically distributed in classical probability theory) random variables that follow a deterministic distribution, and also calculate \eqref{eq:upper_means} and \eqref{eq:lower_means}. In such a case, the gap between the maximum and minimum averages would be small, as observed in Fig. \ref{fig_8},  especially as the window size becomes large. This is because the law of large numbers dictates that the sample means corresponding to the sliding window of 5000 samples shall converge to the true mean of the distribution, as the sample size increases. As a result, the sample means become less variable and more representative of the true mean of the distribution.

\begin{figure}[h]
	\centering
	\includegraphics[scale=0.6]{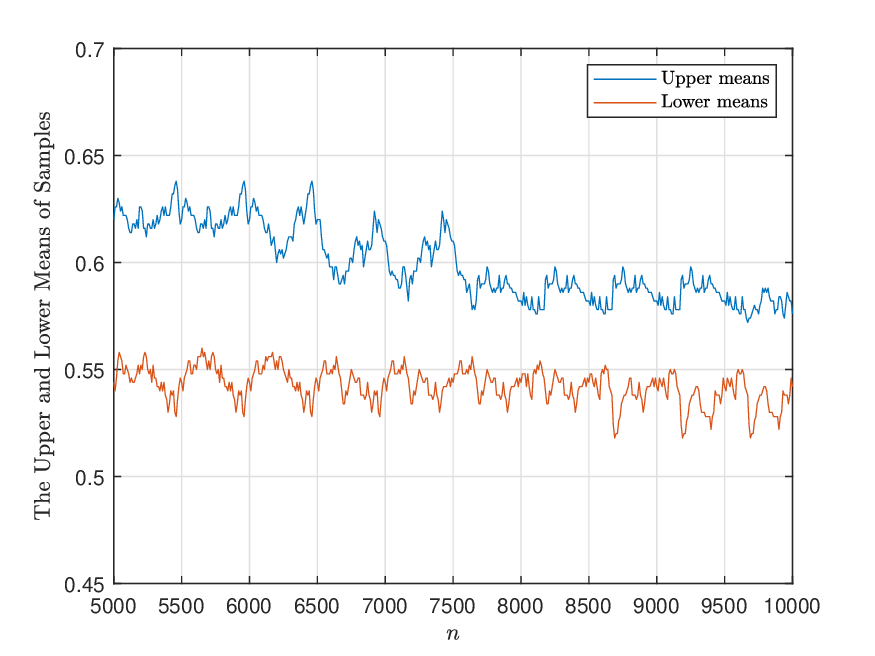}
	\captionsetup{font={scriptsize}}
	\caption{The upper means and lower means of samples of the random variable $X$ with a known Bernoulli distribution $\mathcal{B}(1,0.58)$.}
	\label{fig_8}
\end{figure}

The large gap observed in Fig. \ref{fig_7} suggests that the random variable may not follow a single distribution. Instead, it could exhibit a behavior that is more complex or randomly dynamic, possibly influenced by multiple factors, such as randomly changing conditions over time. This insight is consistent with the idea that the probability model itself may be uncertain or unpredictable, leading to the emergence of the above mentioned second case where the random variable does not follow any specific distribution within the family.

\begin{remark}
	The simulation in this appendix is a heuristic demonstration. In practice, one may take any set of actual data, apply the computations in \eqref{eq:upper_means} and \eqref{eq:lower_means}, and compare the results with those obtained from data generated by a single fixed distribution. A significant gap between the two sets of results would similarly indicate the presence of distribution uncertainty in the actual data.
\end{remark}

The above mentioned second case highlights the limitations of traditional probabilistic approaches when faced with more complex and uncertain system models. In such scenarios, the probability model itself may be uncertain or ill-defined, necessitating a more general and flexible framework. This is where the nonlinear expectation theory comes into play. By relaxing the assumption of a deterministic probability model, the nonlinear expectation theory provides a powerful tool for modeling and analyzing systems with intrinsic model uncertainty. It allows for a holistic consideration of multiple possible distributions, thus valuable for  capturing the ambiguity and unpredictability inherent in complex random systems.

In summary, while traditional probabilistic methods have proven effective in analyzing the above mentioned first case, where the probability model is deterministic but unknown, they fall short in addressing the complexities and uncertainties associated with the second case. The increasing prevalence of complex random systems and the need to model uncertainty more accurately have led to the growing importance of the nonlinear expectation theory. By providing a more general and flexible framework, this theory enables researchers to better understand and handle uncertainty in modern probabilistic modeling.

\section{Proof of Theorem \ref{theentropy}}\label{app2}

For the discrete random variable $X$ with $n$ possible states $\{x_1, \cdots, x_i, \cdots, x_n\}$ defined on a sublinear expectation space $(\Omega,\mathcal{H},\mathbb{E})$, the corresponding uncertain probability distributions of $X$ are $\left\{p_{\theta }\left(X\right)\right\}_{\theta \in \Theta }$. If there is a $\theta_{0}\in\Theta$ with $p_{\theta_0}(x_i)=p_i$ such that $p_{1},\cdots, p_{n}$ are all rational numbers, suppose $p_{i}=\frac{k_{i}}{\sum _{i=1}^{n}k_{i}}$, where $k_{1},\cdots, k_{n}$ are all positive integers, and $i=1, 2, \cdots, n$. In addition, suppose one wants to choose a state from $\sum _{i=1}^{n}k_{i}$ states. Similar to Shannon’s idea, we can divide $\sum _{i=1}^{n}k_{i}$ states into $n$ groups, with $k_{1},\cdots, k_{n}$ states in each group, respectively. For convenience, we also use $k_i$ to denote a specific group. We first randomly choose a group $k_i$ from the groups $k_{1}, \cdots, k_{n}$, and then randomly choose a state from group $k_{i}$. Different from Shannon's methodology, the choice of a single group from the groups $k_{1}, \cdots, k_{n}$ has the same distribution uncertainty with $X$. We use random variable $Y_i$ without distribution uncertainty to describe the state choice from group $k_i$, and the distribution of $Y_i$ is characterized by the uniform distribution of $q=\frac{1}{k_{i}}$. Therefore, on the whole, the states are chosen from $\sum _{i=1}^{n}k_{i}$, where the uniform distribution $\frac{1}{\sum _{i=1}^{n}k_{i}}$ is only a member of the uncertain distribution family for state choice. Thus, we can obtain from Assumptions 2) and 4) in Section \ref{sec3} that

\begin{equation*}
\log \sum _{i=1}^{n}k_{i} \leq \hat{H}\left(X\right)+\sup_{\theta \in \Theta } \sum_{i=1}^{n} p_\theta (x_i)\hat{H}(Y_i)=\hat{H}\left(X\right)+\log \sum _{i=1}^{n}k_{i}+\mathbb{E}[\log p_{\theta_{0}}(X)].
\end{equation*}
Hence,
\begin{equation*}
\hat{H}\left(X\right)\geq -\mathbb{E}[\log p_{\theta_{0}}(X)] \geq \inf_{\theta \in \Theta } \sum _{x}p_{\theta }\left(x\right)\log \frac{1}{p_{\theta }(x)}.
\end{equation*}

If $\hat{H}\left(X\right)<\underset{\theta \in \Theta}{\sup} \sum _{x}p_{\theta }(x)\log \frac{1}{p_{\theta }(x)}$, then there exists some $\epsilon>0$ such that
\begin{equation*}
\hat{H}\left(X\right)+\epsilon<\underset{\theta \in \Theta }{\sup } \sum _{x}p_{\theta }(x)\log \frac{1}{p_{\theta }(x)}.
\end{equation*}
Therefore, there exists a subset $\left\{p_{\theta }\right\}_{\theta \in \Theta^* }$ contained in $\left\{p_{\theta }\left(X\right)\right\}_{\theta \in \Theta }$ such that
\begin{equation*}
\sum _{x}p_{\theta}\left(x\right)\log \frac{1}{p_{\theta}\left(x\right)}>\hat{H}\left(X\right)+\epsilon,\quad \forall \theta \in \Theta^*.
\end{equation*}

Let $Z$ be a random variable for which the uncertain probability distributions are characterized by $\left\{p_{\theta }\right\}_{\theta \in \Theta^* }$. Then, according to Assumption 3), we have
\begin{equation*}
\hat{H}\left(X\right)\geq \hat{H}\left(Z\right)\geq \inf_{\theta \in \Theta^*} \sum _{x}p_{\theta}\left(x\right)\log \frac{1}{p_{\theta}\left(x\right)} \geq \hat{H}\left(X\right)+\epsilon>\hat{H}\left(X\right), 
\end{equation*}
which itself yields a contradiction.

Consequently, $\hat{H}\left(X\right)=\underset{\theta \in \Theta }{\sup } \sum _{x}p_{\theta }(x)\log \frac{1}{p_{\theta }(x)}$.

If some probability distributions in $\left\{p_{\theta }\left(x\right)\right\}_{\theta \in \Theta }$ are incommensurable numbers, they can be approximated by rationals. Then, the same expression must hold due to our continuity assumption, namely Assumption 1).

This concludes the proof of Theorem \ref{theentropy}. $\hfill\blacksquare$

\section{Proof of Theorem \ref{thesource}}\label{app3}

Since $X_{1},X_{2},\cdots,X_{n},\cdots$ is a discrete IID source sequence, for any $i=2,...,n$, the uncertain probability distributions of $X_i$ can also be represented as $\{p_{\theta}(X_1)\}_{\theta \in \Theta}$.

For the statement 1), let $Y_i:=\log \frac{1}{V(X_i)} \in \mathcal{H}$ and $S_n:= \sum_{i=1}^{n}Y_i $. For any $\mu \in [-\mathbb{E}[-Y_i],\mathbb{E}[Y_i]]$ and any $\epsilon>0$, based on the strong law of large numbers under sublinear expectations, there exists a monotonically increasing sequence $\{n_k\}$ such that
\begin{equation*}
V\left( \big|\frac{S_{n_k}}{n_k}-\mu \big| \leq \epsilon \right) \geq 1-\epsilon.
\end{equation*}
Let $\mathcal{M}_{\epsilon}^{(k)}:=\{\boldsymbol{x}^{n_k} \big| |\frac{-\log V(\boldsymbol{x}^{n_k})}{n_k}-\mu| \leq \epsilon \}$. According to \cite{chen2016strong}, we know that for any $X_i,X_j$, there is $V(X_i\in A,X_j\in B)=V(X_i\in A)V(X_j\in B)$. Consequently, $V\left(\mathcal{M}_{\epsilon}^{(k)}\right) \geq 1-\epsilon$, and for any $\boldsymbol{x}^{n_k} \in \mathcal{M}_{\epsilon}^{(k)}$, there is
\begin{equation*}
2^{-n_k (\mu+\epsilon)} \leq V(\boldsymbol{x}^{n_k}) \leq  2^{-n_k (\mu-\epsilon)}.
\end{equation*}

Combining these bounds with the elementary inequalities
\begin{equation*}
1-\epsilon \leq V\left(\mathcal{M}_{\epsilon}^{(k)}\right) \leq \sum_{\boldsymbol{x}^{n_k} \in \mathcal{M}_{\epsilon}^{(k)}} V(\boldsymbol{x}^{n_k}) \leq \| \mathcal{M}_{\epsilon}^{(k)} \|\cdot 2^{-n_k (\mu-\epsilon)},
\end{equation*}
and
\begin{equation*}
\| \mathcal{M}_{\epsilon}^{(k)} \| \cdot 2^{-n_k (\mu+\epsilon)} \leq \sum_{\boldsymbol{x}^{n_k} \in \mathcal{M}_{\epsilon}^{(k)}} V(\boldsymbol{x}^{n_k}) \leq M_n,
\end{equation*}
we obtain
\begin{equation*}
(1-\epsilon)2^{n_k(\mu-\epsilon)} \leq  \| \mathcal{M}_{\epsilon}^{(k)}\| \leq M_n 2^{n_k(\mu+\epsilon)}.
\end{equation*}

Thus, for any $\delta>\epsilon >0$, there exists a sufficiently large $k$ and the source encoding function assigns a unique index to each sequence in $\mathcal{M}_{\epsilon}^{(k)}$. This source encoding function is a one-to-one mapping, whose outputs are easily decodable. Moreover, this nonlinear source code satisfies
\begin{equation*}
\mathcal{E}\left[I_{\{\hat{\boldsymbol{X}}^{n_k}\neq \boldsymbol{X}^{n_k}\}}\right] = v\left((\mathcal{M}_{\epsilon}^{(k)})^{c}\right)=1-V\left(\mathcal{M}_{\epsilon}^{(k)}\right) < \epsilon,
\end{equation*}
and its source coding rate satisfies
\begin{equation*}
R_{\textrm{s}}=\frac{\log \| \mathcal{M}_{\epsilon}^{(k)} \|}{n_k} \leq \frac{\log M_n}{n_k}+\mu+\epsilon \leq \inf_{\theta \in \Theta }\sum_{x}p_{\theta}(x)\log \frac{1}{V(x)} + \delta.
\end{equation*}
By the arbitrariness of $\mu$, $\epsilon$ and $\delta>\epsilon$, the statement 1) is proved.

For the statement 2), let $|\mathcal{X}|=m$. Chosen $\gamma>0$ such that $\hat{H}(X_1)+2\gamma < R_s$, and define
\begin{equation*}
B_n=\left\{\boldsymbol{x}^{n}\in \mathcal{X}^n \Big| H_{\hat{p}_{\boldsymbol{x}^{n}}}\leq \hat{H}(X_1)+\gamma \right\},
\end{equation*}
where $\hat{p}_{\boldsymbol{x}^{n}}$ denotes the empirical distribution (type) of $\boldsymbol{x}^{n}$, and $H_{\hat{p}_{\boldsymbol{x}^{n}}}$ denotes the entropy of distribution $\hat{p}_{\boldsymbol{x}^{n}}$.

For sequences of length $n$ taking values in the alphabet $\mathcal{X}$, the number of possible empirical distributions is at most $(n+1)^m$. Moreover, for any empirical distribution $q$, the number of sequences whose empirical distribution equals $q$ is bounded by $2^{nH_q}$ (see Theorem 11.1.1 and 11.1.3 in \cite{book}). Since every sequence in $B_n$ has an empirical distribution $q$ with $H_q \le \hat{H}(X_1)+\gamma$, it follows that
\begin{equation*}
\|B_n \| \leq (n+1)^m 2^{nH_q} \leq (n+1)^m 2^{n\hat{H}(X_1)+\gamma}.
\end{equation*}

For $n$ sufficiently large, the inequality 
\begin{equation*}
\frac{m\log(n+1)}{n}+\hat{H}(X_1)+\gamma < \hat{H}(X_1)+2\gamma <R_s,
\end{equation*}
holds. Hence, we obtain $\|B_n \|\leq 2^{nR_s}$.

Since the classical entropy is uniformly continuous as a function of the probability distribution, there exists $\delta>0$ such that for any distributions $p,q$, $\|q-p\|_1\leq \delta \implies |H(p)-H(q)| < \gamma$. Consequently, for any $\boldsymbol{x}^{n} \notin B_n$, we have
\begin{equation*}
H_{\hat{p}_{\boldsymbol{x}^{n}}} > \hat{H}(X_1)+\gamma \geq H_{p_\theta} + \gamma \quad \text{for all } \theta\in\Theta.
\end{equation*}
Thus $\|\hat{p}_{\boldsymbol{x}^{n}} - p_\theta\|_1 \ge \delta$ for every $\theta\in\Theta$.

Let
\begin{equation*}
F=\{q| \underset{\theta\in \Theta}{\inf} \| q-p_{\theta} \|_1 \geq \delta \}.
\end{equation*}
Then $(B_n)^c = \{ \boldsymbol{x}^{n}\in \mathcal{X}^n | \hat{p}_{\boldsymbol{x}^{n}}\in F \}$. By the Sanov Theorem under sublinear expectation (Section 3.1 in \cite{fuqing2012relative}), there is
\begin{equation*}
\limsup_{n\rightarrow \infty} \frac{1}{n}\log V((B_n)^c)\leq -\inf_{q\in F}H_{V}(q)<0.
\end{equation*}
Therefore, for any $\epsilon>0$ and sufficiently large $n$, $V((B_n)^c) < \epsilon$.

Thus, for any $\epsilon >0$, there exists a sufficiently large $n$ and the source encoding function assigns a unique index to each sequence in $B_n$. This nonlinear source code satisfies
\begin{equation*}
\mathbb{E}\left[I_{\{\hat{\boldsymbol{X}}^{n}\neq \boldsymbol{X}^{n}\}}\right] = V\left((B_n)^{c}\right) < \epsilon.
\end{equation*}

This concludes the proof of Theorem \ref{thesource}. $\hfill\blacksquare$

\section{Proofs of Theorems in Section \ref{sec4.3}}\label{app4}

\noindent \textbf{Proof of Theorem \ref{thechannel}:}

Suppose the transmission process using the nonlinear channel coding scheme is characterized as
\begin{equation*}
S\xrightarrow{\varphi_n}\boldsymbol{X}^{n}\xrightarrow{\left\{\boldsymbol{P}_{\lambda } \in [0,1]^{\mathcal{X}\times \mathcal{Y}}\right\}_{\lambda \in \Lambda }}\boldsymbol{Y}^{n}\xrightarrow{\psi_n}\hat{S}.
\end{equation*}
Since $S \rightarrow \boldsymbol{X}^{n} \rightarrow \boldsymbol{Y}^{n}$ satisfies the condition in Theorem \ref{thedata}, we have
\begin{equation*}
\overline{I}(S;\boldsymbol{Y}^{n})\leq \overline{I}(\boldsymbol{X}^{n};\boldsymbol{Y}^{n}).
\end{equation*}

For any $\epsilon >0$ and a sequence of $(M,n,\varphi_n,\psi_n)$ nonlinear channel codes with coding rate $R_{\textrm{c}}$ satisfying
\begin{equation*}
\lim _{n\rightarrow \infty } \sup_{\theta \in \Theta} P_{e,\theta}^{(n)} =0, 
\end{equation*}
there exists a sufficiently large $n$ such that $\sup_{\theta \in \Theta} P_{e,\theta}^{(n)} < \epsilon$.

Let us define the random variable $U=I_{\{ \hat{S}\neq S \}}$, then the corresponding uncertain probabilities of the event $\{U=1\}$ is characterized as $\{P_{e,\theta}^{(n)}\}_{\theta \in \Theta}$, and $\mathbb{E}[U]=\underset{\theta \in \Theta}{\sup } P_{e,\theta}^{(n)}$.

Since $\boldsymbol{Y}^{n}$ and $S$ uniquely determine the random variable $U$, we have $\hat{H}\left(U|S,\boldsymbol{Y}^{n}\right)=0$. According to Theorem \ref{theless}, we obtain
\begin{equation}\label{D5}
\hat{H}\left(U,S|\boldsymbol{Y}^{n}\right)=\hat{H}(S|\boldsymbol{Y}^{n})
\end{equation}
and 
\begin{equation}\label{D6}
\hat{H}\left(U,S|\boldsymbol{Y}^{n}\right)\leq \hat{H}\left(U|\boldsymbol{Y}^{n}\right)+\hat{H}(S|U,\boldsymbol{Y}^{n}).
\end{equation}

Consider the term $\hat{H}\left(U|\boldsymbol{Y}^{n}\right)$. Obviously, we have
\begin{equation}\label{D7}
\hat{H}\left(U|\boldsymbol{Y}^{n}\right)\leq 1.
\end{equation}

Furthermore, the term $\hat{H}(S|U,\boldsymbol{Y}^{n})$ satisfies
\begin{equation*}
\hat{H}\left(S|U,\boldsymbol{Y}^{n}\right)=\sup_{\theta\in \Theta} \left[(1-P_{e,\theta}^{(n)})\hat{H}\left(S|\boldsymbol{Y}^{n},U=0\right)+P_{e,\theta}^{(n)} \hat{H}\left(S|\boldsymbol{Y}^{n},U=1\right)\right].
\end{equation*}

Since $\hat{H}\left(S|\boldsymbol{Y}^{n},U=0\right)=0$ and $\hat{H}\left(S|\boldsymbol{Y}^{n},U=1\right)\leq \log \left(\| \mathcal{S} \| -1\right)\leq nR_{\textrm{c}}$, we obtain
\begin{equation}\label{D8}
\hat{H}\left(S|U,\boldsymbol{Y}^{n}\right)\leq \sup_{\theta \in \Theta} P_{e,\theta}^{(n)}n R_{\textrm{c}}.
\end{equation}

Jointly considering \eqref{D5}, \eqref{D6}, \eqref{D7} and \eqref{D8}, we have
\begin{equation}\label{D1}
\hat{H}\left(S|\boldsymbol{Y}^{n}\right)\leq \hat{H}\left(U|\boldsymbol{Y}^{n}\right)+\hat{H}\left(S|U,\boldsymbol{Y}^{n}\right)\leq 1+\sup_{\theta \in \Theta} P_{e,\theta}^{(n)} n R_{\textrm{c}}.
\end{equation}

Suppose that the uncertain probability distributions of the message $S\in \mathcal{S}=\{1,2,\cdots,M\}$ to be sent is denoted as $\{p_{\sigma}(S)\}_{\sigma \in \Sigma}$, and the uniform distribution $p\left(s\right)=\frac{1}{M}$ is contained in the family of uncertain probability distributions $\{p_{\sigma}(S)\}_{\sigma \in \Sigma}$. Then, we have
\begin{equation}\label{D2}
n R_{\textrm{c}}=\hat{H}\left(S\right)\leq \hat{H}\left(S|\boldsymbol{Y}^{n}\right)+\overline{I}\left(S;\boldsymbol{Y}^{n}\right)\leq \hat{H}\left(S|\boldsymbol{Y}^{n}\right)+\overline{I}\left(\boldsymbol{X}^{n};\boldsymbol{Y}^{n}\right).
\end{equation}

Let $p_{\theta_1,\cdots,\theta_n,\lambda}(\boldsymbol{x}^{n},\boldsymbol{y}^{n})=\prod_{i=1}^{n}p_{\theta_i}(x_i)\cdot \prod_{i=1}^{n}p_{\lambda}(y_i|x_i)$, $p_{\theta_1,\cdots,\theta_n}(\boldsymbol{x}^{n})=\prod_{i=1}^{n}p_{\theta_i}(x_i)$, $p_{\theta_1,\cdots,\theta_n,\lambda}(\boldsymbol{y}^{n})=\prod_{i=1}^{n}(\underset{x_j\in \mathcal{X}}{\sum}p_{\theta_i}(x_j)p_{\lambda}(y_i|x_j))$, then we have
\begin{equation}\label{D3}
\begin{split}
\overline{I}\left(\boldsymbol{X}^{n};\boldsymbol{Y}^{n}\right)=&\sup_{\theta_1,\cdots,\theta_n \in \Theta} \sup_{\lambda \in \Lambda} \sum _{\boldsymbol{x}^{n},\boldsymbol{y}^{n}}p_{\theta_1,\cdots,\theta_n,\lambda}(\boldsymbol{x}^{n},\boldsymbol{y}^{n})\log \frac{p_{\theta_1,\cdots,\theta_n,\lambda}(\boldsymbol{x}^{n},\boldsymbol{y}^{n})}{p_{\theta_1,\cdots,\theta_n}(\boldsymbol{x}^{n}) p_{\theta_1,\cdots,\theta_n,\lambda}(\boldsymbol{y}^{n})}\\
\leq & \sup_{\theta_1,\cdots,\theta_n \in \Theta} \sup_{\lambda \in \Lambda} \sum _{i=1}^{n}\sum _{x_{i},y_{i}}p_{\theta_i }\left(x_{i}\right)p_{\lambda}(y_{i}|x_{i})\log \frac{p_{\theta_i,\lambda}(x_i,y_i)}{p_{\theta_i}(x_i)p_{\theta_i,\lambda}(y_i)}\\
\leq & \sum _{i=1}^{n}\sup_{\theta_i \in \Theta} \sup_{\lambda \in \Lambda} \sum _{x_{i},y_{i}}p_{\theta_i }\left(x_{i}\right)p_{\lambda}(y_{i}|x_{i})\log \frac{p_{\theta_i,\lambda}(x_i,y_i)}{p_{\theta_i}(x_i)p_{\theta_i,\lambda}(y_i)} \leq n\overline{C}.
\end{split}
\end{equation}

Substitute Eqs. (\ref{D1}) and (\ref{D3}) into (\ref{D2}), we have
\begin{equation*}
n R_{\textrm{c}}\leq 1+\sup_{\theta \in \Theta} P_{e,\theta}^{(n)}  n R_{\textrm{c}} + n\overline{C}.
\end{equation*}
Dividing by $n$, there is
\begin{equation}\label{D4}
R_{\textrm{c}} \leq \frac{1}{n}+\sup_{\theta \in \Theta} P_{e,\theta}^{(n)} R_{\textrm{c}}+\overline{C}.
\end{equation}
The above inequality (\ref{D4}) holds for any sufficiently large $n$. Note that as $n\rightarrow \infty $, both $\frac{1}{n}$ and $\underset{\theta \in \Theta}{\sup } P_{e,\theta}^{(n)}$ tend to 0. Therefore, we have $R_{\textrm{c}} \leq \overline{C}$.

This concludes the proof of Theorem \ref{thechannel}. $\hfill\blacksquare$

\noindent \textbf{Proof of Theorem \ref{thechannel1}:}

Since $R_{\textrm{c}}<\overline{C}$, there exists a constant $R'$ with $R_{\textrm{c}}<R'<\overline{C}$ and a probability distribution $p$ on $\mathcal{X}$ such that $\underset{\lambda \in \Lambda}{\sup} \underset{x\in\mathcal{X},y\in\mathcal{Y}}{\sum} p\left(x\right)p_{\lambda }\left(y|x\right)\log \frac{p_{\lambda}\left(y|x\right)}{\sum_{x\in \mathcal{X}}p\left(x\right)p_{\lambda}\left(y|x\right)}>R'$.

Let $p(\boldsymbol{x}^{n})=\prod_{i=1}^{n}p(x_i)$ and $\mathcal{G}^n=\left\{(\boldsymbol{x}^{n},\boldsymbol{y}^{n})\Big| \underset{\lambda \in \Lambda}{\sup } \, \log \frac{p_{\lambda}\left(\boldsymbol{y}^n|\boldsymbol{x}^n\right)}{\sum_{\boldsymbol{x}^n \in \mathcal{X}^{n}}p(\boldsymbol{x}^n)p_{\lambda}\left(\boldsymbol{y}^n|\boldsymbol{x}^n\right)}\geq nR'\right\}$.

For the codebook $\{\boldsymbol{x}_{1},\boldsymbol{x}_{2},\cdots ,\boldsymbol{x}_{{\|\mathcal{X}\|^{n R_c}}}\}$ with the blocklength $n$, we set the decoding rule as follows: Suppose the uncertain-distribution channel $[\mathcal{X},\left\{\boldsymbol{P}_{\lambda } \in [0,1]^{\mathcal{X}\times \mathcal{Y}}\right\}_{\lambda \in \Lambda },\mathcal{Y} ]$ processes the transmitted codeword $\boldsymbol{x}_{i}$ and produces an output sequence $\boldsymbol{y}^{n}$. Then the receiver considers the set $\mathcal{G}^{*}\left(\boldsymbol{y}^{n}\right)=\left\{\boldsymbol{x}^{n}\big|\left(\boldsymbol{x}^{n},\boldsymbol{y}^{n}\right)\in \mathcal{G}^{n}\right\}$. If there is only a single element $\hat{\boldsymbol{x}}^{n}$ in the set $\mathcal{G}^{*}\left(\boldsymbol{y}^{n}\right)$, the decoding function is applied to $\boldsymbol{y}^{n}$ to estimate the original message as this element. If there are more than one element in the set $\mathcal{G}^{*}\left(\boldsymbol{y}^{n}\right)$, an error occurrs. For this nonlinear channel coding scheme, the error event for codeword $\boldsymbol{x}_i$ is contained in
\begin{equation*}
\{ \boldsymbol{x}_{i} \text{ error} \} \subset \{ \boldsymbol{x}_{i} \notin \mathcal{G}^{*}\left(\boldsymbol{y}^{n}\right) \} \cup \left( \underset{j \neq i}{\cup}\{ \boldsymbol{x}_{j} \in \mathcal{G}^{*}\left(\boldsymbol{y}^{n}\right) \} \right).
\end{equation*}

Using the random coding technique and taking the lower expectation $\mathcal{E}$ on both sides of the above inequality and using the sub-additivity of capacities, we have
\begin{align*}
&\mathcal{E}\left[I_{\{\hat{\boldsymbol{x}}^{n}\neq \boldsymbol{x}^{n}\}}\right] \leq v\left( \left\{ \boldsymbol{X}^{n}\notin \mathcal{G}^{*}({\boldsymbol Y}^{n}) \right\} \cup \left( \underset{\tilde{\boldsymbol{X}}^{n}\neq \boldsymbol{X}^{n}}{\cup} \left\{ \tilde{\boldsymbol{X}}^{n}\in \mathcal{G}^{*}(\boldsymbol{Y}^{n}) \right\} \right) \right) \\
\leq & v\left(\boldsymbol{X}^{n}\notin \mathcal{G}^{*}({\boldsymbol Y}^{n})\right)+\sum _{\tilde{\boldsymbol{X}}^{n}\neq \boldsymbol{X}^{n}} V \left(\tilde{\boldsymbol{X}}^{n}\in \mathcal{G}^{*}(\boldsymbol{Y}^{n})\right)\\
:=&I_{1}+I_{2}.
\end{align*}

For the term $I_{1}$, there is
\begin{equation*}
I_1 = v\left( (\boldsymbol{X}^{n},\boldsymbol{Y}^{n})\notin \mathcal{G}^{n} \right) = v\left( \underset{\lambda \in \Lambda}{\sup } \, \log \frac{p_{\lambda}\left(\boldsymbol{Y}^n|\boldsymbol{X}^n\right)}{\sum_{\boldsymbol{x}^n \in \mathcal{X}^{n}}p(\boldsymbol{x}^n)p_{\lambda}\left(\boldsymbol{Y}^n|\boldsymbol{x}^n\right)} < nR' \right).
\end{equation*}
Due to $R'<\overline{C}$, by the definition of $\overline{C}$, there exist $\lambda \in \Lambda$ such that the upper expectation of $\log \frac{p_{\lambda}\left(Y|X\right)}{\sum_{x\in \mathcal{X}}p\left(x\right)p_{\lambda}\left(Y|x\right)}$ is greater than $R'$. Using the strong law of large numbers under sublinear expectations, there exists a sufficiently large $n$ such that
\begin{equation*}
V\left(\underset{\lambda \in \Lambda}{\sup } \, \log \frac{p_{\lambda}\left(\boldsymbol{Y}^n|\boldsymbol{X}^n\right)}{\sum_{\boldsymbol{x}^n \in \mathcal{X}^{n}}p(\boldsymbol{x}^n)p_{\lambda}\left(\boldsymbol{Y}^n|\boldsymbol{x}^n\right)} \geq nR' \right) > 1-\frac{\epsilon}{2}.
\end{equation*}
Therefore, $I_1 <\frac{\epsilon}{2}$.

For the term $I_{2}$, for any fixed erroneous codeword $\tilde{\boldsymbol{X}}^n$ (independent of $\boldsymbol{Y}^n$ in the random coding average), the event $\{\tilde{\boldsymbol{X}}^n \in \mathcal{G}^{*}(\boldsymbol{Y}^n)\}$ is equivalent to
\begin{equation*}
\left\{ \sup_{\lambda\in \Lambda} \frac{1}{n}\sum_{i=1}^n \log\frac{p_{\lambda}(Y_i|\tilde{X}_i)}{\sum_{x \in \mathcal{X}} p(x) p_{\lambda}(Y_i|x)}\geq R' \right\}.
\end{equation*}
Define
\begin{equation*}
S_n^{\lambda,p} := \frac{1}{n}\sum_{i=1}^n \log \frac{p_{\lambda}(Y_i|\tilde{X}_i)}	{\sum_{x} p(x) p_{\lambda}(Y_i|x)}.
\end{equation*}
By the large deviation under sublinear expectations (Theorem 3.1.(I) in \cite{tan2020large}), there exists sufficiently large $n$ and constant $\alpha > 0$ such that for sufficiently large $n$,
\begin{equation}\label{eq1}
\sup_{\lambda \in \Lambda}
V\left( S_n^{\lambda,p} \geq R' - \delta \right)
\le \|\mathcal{X}\|^{-n\alpha}.
\end{equation}

Since $\{P_{\theta}\}_{\theta\in \Theta}$ is a countably-dimensional weakly compact family, there exists a finite $\lambda$-net with covering number $ N_\lambda < \infty$, independent of $n$, such that
\begin{equation*}
V\left( \tilde{\boldsymbol{X}}^n \in \mathcal{G}^{*}(\boldsymbol{Y}^n) \right) \leq \sum_{l=1}^{N_\lambda} 	V\left( S_n^{\lambda_l, p} \geq R' - \delta \right).
\end{equation*}
Applying (\ref{eq1}) yields
\begin{equation}\label{eq2}
V\left( \tilde{\boldsymbol{X}}^n \in \mathcal{G}^{*}(\boldsymbol{Y}^n) \right)	\leq	N_\lambda \cdot \|\mathcal{X}\|^{-n\alpha}.
\end{equation}

Since there are $M \le \|\mathcal{X}\|^{n R_c}$ codewords in the codebook, substituting (\ref{eq2}) into the definition of $I_2$ gives
\begin{equation*}
I_2 \leq \|\mathcal{X}\|^{n R_c} \cdot N_\delta \cdot \|\mathcal{X}\|^{-n\alpha} = N_\delta \cdot \|\mathcal{X}\|^{-n(\alpha-R_c)}.
\end{equation*}

By choosing $\delta$ sufficiently small, we can ensure that $R'-\delta>R_c$ and then calculate $\alpha$ to find $\alpha > R'-\delta$, it follows that $I_2 \to 0$ as $n \to \infty$. Therefore, for sufficiently large $n$, we have $I_2 < \frac{\epsilon}{2}$. Consequently, we have $\mathcal{E}\left[I_{\{\hat{\boldsymbol{x}}^{n}\neq \boldsymbol{x}^{n}\}}\right] <\epsilon $.

This concludes the proof of Theorem \ref{thechannel1}. $\hfill\blacksquare$

\section{Proofs of Theorems in Section \ref{sec4.4}}\label{app5}

\noindent  \textbf{Proof of Theorem \ref{thelim1}:}

\begin{enumerate}
	\item {This property is obvious.}
	
	\item {$\forall D_{1},D_{2}\geq 0$, $\lambda \in (0,1)$, $\forall \epsilon >0$, there exists transition probability matrices $ \boldsymbol{Q}_{1}$ and $\boldsymbol{Q}_{2}$ such that
		\begin{equation*}
		\overline{I}\left[\{p_{\theta}(X)\}_{\theta \in \Theta};\boldsymbol{Q}_1\left(\hat{X}|X\right)\right]\leq \hat{R}^{I}\left(D_{1}\right)+\epsilon,
		\end{equation*}
		\begin{equation*}
		\overline{I}\left[\{p_{\theta}(X)\}_{\theta \in \Theta};\boldsymbol{Q}_2\left(\hat{X}|X\right)\right]\leq \hat{R}^{I}\left(D_{2}\right)+\epsilon.
		\end{equation*}
		
		Due to $\mathbb{E}_{\lambda {\boldsymbol{Q}_{1}}+\left(1-\lambda \right){\boldsymbol{Q}_{2}}}\left[d\left(X,\hat{X}\right)\right]\leq \lambda D_{1}+(1-\lambda )D_{2}$, we can get
		\begin{equation*}
		\begin{split}
		&\hat{R}^{I}\left(\lambda D_{1}+(1-\lambda )D_{2}\right)\\
		\leq& \overline{I}\left[\{p_{\theta}(X)\}_{\theta \in \Theta};\left(\lambda \boldsymbol{Q}_{1}\left(\hat{X}|X\right)+\left(1-\lambda \right)\boldsymbol{Q}_{2}\left(\hat{X}|X\right)\right)\right] \\
		\leq& \lambda\overline{I}\left[\{p_{\theta}(X)\}_{\theta \in \Theta};\boldsymbol{Q}_1\left(\hat{X}|X\right)\right]+\left(1-\lambda \right)\overline{I}\left[\{p_{\theta}(X)\}_{\theta \in \Theta};\boldsymbol{Q}_2\left(\hat{X}|X\right)\right]\\
		\leq& \lambda \hat{R}^{I}\left(D_{1}\right)+\left(1-\lambda \right)\hat{R}^{I}\left(D_{2}\right)+\epsilon.
		\end{split}
		\end{equation*}
		
		Since $\epsilon $ is an arbitrary positive value, $\hat{R}^{I}\left(D\right)$ is a convex function.$\hfill\blacksquare$}
\end{enumerate}

\noindent  \textbf{Proof of Theorem \ref{thelim2}:}

Let the transition probability matrix $\boldsymbol{Q}(\hat{X}|X)$ be the one that attains the infimum of $\hat{R}^{I}\left(D\right)$. Because of $R_{\textrm{s}}>\hat{R}^{I}\left(D\right)$, there exists a constant $R'$ such that $R_{\textrm{s}}>R'> \hat{R}^{I}\left(D\right)$.

Let $\mathcal{L}^{n}=\{ (\boldsymbol{x}^{n},\hat{\boldsymbol{x}}^{n})\big| \underset{\theta \in \Theta}{\inf } \, \log \frac{q(\hat{\boldsymbol{x}}^n|\boldsymbol{x}^n)}{\sum_{\boldsymbol{x}^n \in \mathcal{X}^n}q(\hat{\boldsymbol{x}}^n|\boldsymbol{x}^n)p_{\theta}(\boldsymbol{x}^n)}<nR' \}$. For any $\epsilon>0$, let $\mathcal{B}_{\epsilon}^{n} = \{ (\boldsymbol{x}^{n},\hat{\boldsymbol{x}}^{n})\big| \break  |d(\boldsymbol{x}^{n}, \hat{\boldsymbol{x}}^{n})-\mathbb{E}_{\boldsymbol{Q}}[d(X,\hat{X})]|<\frac{\epsilon}{2} \}$.

Based on the proof of Theorem \ref{thesource}, for any $\epsilon'>0$, there exists a set $\mathcal{M}_{\epsilon'}^{(k)}$ satisfying $v\left((\mathcal{M}_{\epsilon' }^{\left(k\right)})^{c}\right)\leq \frac{\epsilon' }{3}$. Let $\mathcal{L}^*=\{(\boldsymbol{x}^{n_k},\hat{\boldsymbol{x}}^{n_k})\big|(\boldsymbol{x}^{n_k},\hat{\boldsymbol{x}}^{n_k})\in \mathcal{B}_{\epsilon}^{n_k}\cap \mathcal{L}^{n_k}, \ \boldsymbol{x}^{n_k} \in \mathcal{M}_{\epsilon'}^{(k)} \}$. For the vector $\boldsymbol{x}^{n_k}$, find an index $W\in \{1,2,\cdots, \|\mathcal{W}\| \}$ such that $(\boldsymbol{x}^{n_k},\hat{\boldsymbol{x}}^{n_k}(W))\in \mathcal{L}^*$. If there are multiple such indices, choose the smallest one among them. If there is no such index, set $W=1$. The decoder makes the estimated sequence as $\hat{\boldsymbol{x}}^{n_k} = \hat{\boldsymbol{x}}^{n_k}(W)$.

Let $d_{\max}=\underset{x,\hat{x}\in \mathcal{X}}{\max} \ d(x,\hat{x})$, there holds
\begin{equation*}
\mathcal{E}_{\boldsymbol{Q}}\left[d\left(\boldsymbol{X}^{n_k},\hat{\boldsymbol{X}}^{n_k}\right)\right] \leq D+\frac{\epsilon}{2}+\mathcal{E}_{\boldsymbol{Q}}\left[I_{ \{ (\boldsymbol{X}^{n_k},\hat{\boldsymbol{X}}^{n_k})\notin \mathcal{L}^* \} }\right] d_{\max}.
\end{equation*}

Consider the term $\mathcal{E}_{\boldsymbol{Q}}\left[I_{ \{ (\boldsymbol{X}^{n_k},\hat{\boldsymbol{X}}^{n_k})\notin \mathcal{L}^* \} }\right]$, we have
\begin{equation*}
\begin{split}
\mathcal{E}_{\boldsymbol{Q}}\left[I_{ \{ (\boldsymbol{X}^{n_k},\hat{\boldsymbol{X}}^{n_k})\notin \mathcal{L}^* \} }\right]  &\leq \mathbb{E}_{\boldsymbol{Q}}\left[\prod_{\omega=1}^{2^{nR_c}} V((\boldsymbol{x}^{n_k},\hat{\boldsymbol{X}}^{n_k})\notin \mathcal{L}^*)|\boldsymbol{x}^{n_k}=\boldsymbol{X}^{n_k}\right]+\frac{\epsilon'}{3}\\
&\leq \mathbb{E}_{\boldsymbol{Q}}\left[ [1-\underset{\hat{\boldsymbol{x}}^{n_k}}{\sum}q(\hat{\boldsymbol{x}}^{n_k}|\boldsymbol{x}^{n_k})2^{-n_k R'}]^{2^{n_k R_{\textrm{s}}}}|\boldsymbol{x}^{n_k}=\boldsymbol{X}^{n_k}\right]+\frac{\epsilon'}{3}\\
&\leq \mathbb{E}_{\boldsymbol{Q}}\left[e^{-2^{-n_k R'}2^{n_k R_{\textrm{s}}}}\right] +\frac{\epsilon'}{3}\\
&=e^{-2^{n_k(R_{\textrm{s}}-R')}}+\frac{\epsilon'}{3}.
\end{split}
\end{equation*}
Note that we have $R'<R_{\textrm{s}}$, hence there exists a sufficiently large $n_k$ and a sufficiently small $\epsilon'$ such that $\mathcal{E}_{\boldsymbol{Q}}\left[I_{ \{ (\boldsymbol{X}^{n_k},\hat{\boldsymbol{X}}^{n_k})\notin \mathcal{L}^* \} }\right]\leq \frac{\epsilon}{2d_{\max}}$. Therefore, we have $ \mathcal{E}_{\boldsymbol{Q}}\left[d(\boldsymbol{X}^{n_k},\hat{\boldsymbol{X}}^{n_k})\right] \leq D+\epsilon$.

This concludes the proof of Theorem \ref{thelim2}. $\hfill\blacksquare$

\end{document}